%% file: paper.tex
\documentclass[twoside,leqno,twocolumn]{article}

\usepackage[letterpaper]{geometry}

\usepackage{siamproceedings}

\input{preamble}

\title{Zip-Tries: Simple Dynamic Data Structures for Strings}

\author{David Eppstein\thanks{Department of Computer Science, The University of California, Irvine, USA (\texttt{eppstein@uci.edu}, \texttt{ogila@uci.edu}, \texttt{goodrich@uci.edu}, \texttt{ryutok@uci.edu}).}
\and Ofek Gila\footnotemark[1]
\and Michael T. Goodrich\footnotemark[1]
\and Ryuto Kitagawa\footnotemark[1]}
\date{}

\begin{document}

\setcounter{page}{0}
\maketitle

\fancyfoot[R]{\scriptsize{Copyright \textcopyright\ 2025 by SIAM\\
Unauthorized reproduction of this article is prohibited}}
\input{abstract}

\clearpage
\input{intro}
\input{paradigm}
\input{zip-tree}
\input{zip-trie}
\input{string-b-tree}
\input{experiments}
\input{future}

\clearpage
\bibliography{paper}

\clearpage
\appendix

\input{related-work}
\input{string-comparisons}
\input{more-experiments}
% \clearpage
\input{omitted-proofs}

\end{document}

%% file: preamble.tex
\usepackage{amsmath}

\usepackage{tikz}
\usepackage{xcolor}
\usetikzlibrary{trees,positioning,decorations.pathreplacing,calc,patterns,fit,shapes,backgrounds,arrows.meta}
\usepackage{color}
\usepackage{graphicx}
\usepackage{calc}
\usepackage{xstring}
\usepackage{url}
\usepackage{ifthen}
\usepackage[misc]{ifsym}
\usepackage{pifont}
\usepackage[utf8]{inputenc}
\usepackage[T1]{fontenc}
\usepackage{lmodern}
\usepackage{microtype} % against underfull hbox
\usepackage{dblfloatfix}
\usepackage{enumitem}
\usepackage{multirow}
\usepackage{algpseudocodex}
\usepackage{changepage}
\usepackage{silence} % suppress slanted font warning which is irrelevant

\usepackage{amsfonts}
\usepackage{epstopdf}
\ifpdf
  \DeclareGraphicsExtensions{.eps,.pdf,.png,.jpg}
\else
  \DeclareGraphicsExtensions{.eps}
\fi

\usepackage{amsopn}

\newsiamremark{remark}{Remark}
\newsiamremark{hypothesis}{Hypothesis}
\crefname{hypothesis}{Hypothesis}{Hypotheses}
\newsiamthm{claim}{Claim}
\newsiamthm{fact}{Fact}

\usepackage[]{todonotes} %disable option

\newcommand{\IF}{\textbf{if}~}
\newcommand{\THEN}{~\textbf{then}~}

\crefname{fact}{fact}{facts}
\Crefname{fact}{Fact}{Facts}

\hypersetup{
   colorlinks=true,
   linkcolor=blue,
   urlcolor=blue,
   linktoc=all,
   citecolor=blue
}

\definecolor{okabe1}{HTML}{000000}
\definecolor{okabe2}{HTML}{E69F00}
\definecolor{okabe3}{HTML}{56B4E9}
\definecolor{okabe4}{HTML}{009E73}
\definecolor{okabe5}{HTML}{F0E442}
\definecolor{okabe6}{HTML}{0072B2}
\definecolor{okabe7}{HTML}{D55E00}
\definecolor{okabe8}{HTML}{CC79A7}

\newcommand{\cmark}{\textcolor{okabe4}{\ding{51}}}%
\newcommand{\xmark}{\textcolor{okabe7}{\ding{55}}}%

\renewcommand{\emph}[1]{\textit{\textbf{#1}}}
\let\epsilon\varepsilon

\newcommand{\verytiny}{\fontsize{4pt}{5pt}\selectfont}

\newcommand{\keyvalue}[2]{%
    \IfStrEq{#1}{.}{$-\infty$}{%  Check if #1 is A
        \IfStrEq{#1}{,}{$\infty$}{% Check if #1 is Z
            \ifnum#2=-1
                #1%
            \else
                \pgfmathtruncatemacro{\next}{#2+1}%
                \textbf{\textcolor{okabe4}{\StrLeft{#1}{#2}}}%
                \textbf{\textcolor{okabe7}{\StrChar{#1}{\next}}}%
                \StrGobbleLeft{#1}{\next}%
            \fi
        }%
    }%    
}

\makeatletter
\newcommand{\lcp}[3][]{%
    \StrCompare{#2}{#3}[\lcp@result]%
    \ifnum\lcp@result=0%
        \StrLen{#2}[\lcp@length]%
        \def\lcp@final{\lcp@length}%
    \else%
        \def\lcp@final{\number\numexpr\lcp@result-1\relax}%
    \fi%
    \if\relax\detokenize{#1}\relax
        \lcp@final%
    \else%
        \@ifundefined{c@#1}{%
            \newcounter{#1}%
            \setcounter{#1}{\lcp@final}%
        }{%
            \setcounter{#1}{\lcp@final}%
        }%
    \fi%
}
\makeatother

\newcommand{\drawReducedWidthBrace}[3]{
    \ifthenelse{\equal{#2}{above}}{
        \draw[decorate,decoration={brace,amplitude=5pt}] 
        ([xshift=5pt]#1.north west) -- ([xshift=-5pt]#1.north east) 
        node[midway,above=5pt] (#1-above) {#3};
    }{
        \draw[decorate,decoration={brace,amplitude=5pt,mirror}] 
        ([xshift=5pt]#1.south west) -- ([xshift=-5pt]#1.south east) 
        node[midway,below=5pt] (#1-below) {#3};
    }
}

\newcommand{\convexpath}[2]{
[   
    create hullnodes/.code={
        \global\edef\namelist{#1}
        \foreach [count=\counter] \nodename in \namelist {
            \global\edef\numberofnodes{\counter}
            \node at (\nodename) [draw=none,name=hullnode\counter] {};
        }
        \node at (hullnode\numberofnodes) [name=hullnode0,draw=none] {};
        \pgfmathtruncatemacro\lastnumber{\numberofnodes+1}
        \node at (hullnode1) [name=hullnode\lastnumber,draw=none] {};
    },
    create hullnodes
]
($(hullnode1)!#2!-90:(hullnode0)$)
\foreach [
    evaluate=\currentnode as \previousnode using \currentnode-1,
    evaluate=\currentnode as \nextnode using \currentnode+1
    ] \currentnode in {1,...,\numberofnodes} {
  let
    \p1 = ($(hullnode\currentnode)!#2!-90:(hullnode\previousnode)$),
    \p2 = ($(hullnode\currentnode)!#2!90:(hullnode\nextnode)$),
    \p3 = ($(\p1) - (hullnode\currentnode)$),
    \n1 = {atan2(\y3,\x3)},
    \p4 = ($(\p2) - (hullnode\currentnode)$),
    \n2 = {atan2(\y4,\x4)},
    \n{delta} = {-Mod(\n1-\n2,360)}
  in 
    {-- (\p1) arc[start angle=\n1, delta angle=\n{delta}, radius=#2] -- (\p2)}
}
-- cycle
}

%% file: abstract.tex
\begin{abstract}
% In this paper, we introduce \emph{zip-tries}, which are simple dynamic string
% data structure.
% We study both sequential and parallel implementations of zip-tries, showing how
% they can match or improve previous bounds for searching and updating a dynamic
% dictionary of character strings over a finite alphabet.

% \ofek{Our sequential zip-tries do not match the theoretical bounds of all other data structures. They are more lcp-aware/cache friendly than most. The main theoretical results are the parallel results, specifically with our modifications of the string B-tree. No other implementation comes close to our PEM results.}
% \ofek{Included some math results. there are many more I didn't include yet}

In this paper, we introduce \emph{zip-tries}, which are simple, dynamic,
memory-efficient data structures for strings.
Zip-tries support search and update operations for $k$-length strings in
$\mathcal{O}(k+\log n)$ time in the standard RAM model or in
$\mathcal{O}(k/\alpha+\log n)$ time in the word RAM model, where $\alpha$ is the
length of the longest string that can fit in a memory word,
and $n$ is the number of strings in the trie.
Importantly, we show how zip-tries can achieve this while only requiring
$\mathcal{O}(\log{\log{n}} + \log{\log{\frac{k}{\alpha}}})$ bits of metadata per node w.h.p.%
, which is an exponential improvement over previous results for long strings.
% \ofek{Actually, the amount of metadata spent on the LCP lengths is proportional to $\ell/\alpha$ which may be much less than $k/\alpha$. This nuance is present in the main text, but I felt it was overkill to include here. Feel free to change the wording here if you think of a nice way to word it.}
Despite being considerably simpler and more memory efficient, we show how
zip-tries perform competitively with state-of-the-art data structures on large
datasets of long strings.

Furthermore, we provide a simple, general framework for parallelizing string
comparison operations in linked data structures, which we apply to zip-tries to
obtain \emph{parallel zip-tries}.
Parallel zip-tries are able to achieve good search and update performance in
parallel, performing such operations in $\mathcal{O}(\log{n})$ span.
We also apply our techniques to an existing external-memory string data
structure, the \textit{string B-tree}, obtaining a \emph{parallel string B-tree}
which performs search operations using $\mathcal{O}(\log_B{n})$ I/O span and
$\mathcal{O}(\frac{k}{\alpha B} + \log_B{n})$ I/O work in the parallel external
memory (PEM) model.
The parallel string B-tree can perform prefix searches using only
$\mathcal{O}(\frac{\log{n}}{\log{\log{n}}})$ span under the practical PRAM
model.

For the case of long strings that share short common prefixes, we provide
\textit{LCP-aware} variants of all our algorithms that should be quite efficient
in practice, which we justify empirically.
\end{abstract}

%% file: intro.tex
% !TEX root = paper.tex
\section{Introduction}

A \emph{trie}, or \emph{digital tree},
is a classical tree structure designed for storing and performing
operations on a collection of character strings, known as a ``dictionary.'' 
Each string is represented as a path from
the root to a leaf, with each edge representing a character in the string. 
There are numerous trie variants, including the well-known
Patricia trie~\cite{PATRICIA}, 
where every non-branching path is
compressed into a single edge. 
Tries are simple structures used extensively in 
practice~\cite{gonnetExhaustiveMatchingEntire1992,chenTriebasedDataStructures1997,wonEfficientApproachSequence2006,holleyBloomFilterTrie2016},
but they are neither
well-suited for large data sets with long strings in external memory nor for use in parallel algorithms.
Thus, it would be of interest to have a string data structure that is simple and practical
while also being adaptable to external-memory and parallel settings.

While strings may be long and string dictionaries can be large,
the size of the alphabet, $\sigma,$ is often small,
with ASCII containing only 128 characters and DNA sequences containing only 4;
hence we assume $\sigma$ is a fixed constant. 
The \emph{word RAM} and \emph{practical RAM} 
models~\cite{miltersenLowerBoundsStatic1996}
take advantage of this by packing $\alpha = \bigl\lfloor w / \lceil\log_2{\sigma}\rceil\bigr\rfloor $ 
characters into a machine word, where $w$ is the word size in bits. 
The word RAM allows
AC$^0$ operations to be performed on machine words in $\mathcal{O}(1)$ time,
and the practical RAM model allows constant-time operations that are common in modern
CPU machine languages, such as bitwise-AND, bitwise-OR, 
and most significant bit (MSB), e.g., see~\cite{thorup-ac0}.
Both models allow comparisons between two strings of length $k$ to be performed in
$\mathcal{O}(k/\alpha)$ time. 
While traditional trie variants such as the
Patricia trie do not take advantage of this, more recent 
variants do~\cite{belazzouguiDynamicZfastTries2010,DBLP:journals/ieicet/TakagiISA17,tsurutaCtrieDynamicTrie2022},
by having each node represent machine words rather than individual characters,
and incorporating hash tables or dynamic predecessor/successor data
structures, such as z-fast tries~\cite{belazzouguiDynamicZfastTries2010},
to each node. 
Although these structures achieve
good theoretical bounds, they are difficult to implement and,
due to their sequential nature,
are not easily parallelizable.

% There are dozens of trie variants, with one recent result able to perform
% prefix search and update operations in $\mathcal{O}(k/\alpha + \log{\log{N}})$
% time, where $N$ is the total length of all the strings in the trie.

Another well-known data structure, typically designed for
numerical keys or keys that are fixed-length numerical tuples,
is the skip list, which stores a sorted collection of items to support
efficient search, insertion, and deletion operations.
One of the key benefits of skip lists over similar data
structures is its inherent 
simplicity~\cite{pughSkipListsProbabilistic1990,munroDeterministicSkipLists1992},
leading to their wide use in practice.
% not sure if this claim requires a citation, there are many
% Furthermore, skip lists only branch $\mathcal{O}(\log{n})$ 
% times per operation,
% lending themselves to be more easily parallelizable. 
Unfortunately, skip lists
are inefficient for strings, since comparisons between strings of
length $k$ takes time $\mathcal{O}(k)$, or $\mathcal{O}(k/\alpha)$ in
the word RAM model. Search and update operations on a skip list perform
$\mathcal{O}(\log{n})$ comparisons (w.h.p.~in randomized 
versions)~\cite{pughSkipListsProbabilistic1990,munroDeterministicSkipLists1992},
leading to a total time complexity of
$\mathcal{O}((k/\alpha) \log{n})$ when applied ``off the shelf'' to 
string keys. The logarithmic factor in this bound negates the advantages of
the word RAM model for long strings under realistic values of $\alpha$.

There has been some work adapting skip lists for strings in the external memory
model---where the goal is to minimize I/O operations when the data
structure exceeds main memory---notably with the self-adjusting
skip list (SASL) by Ciriani et al.~\cite{cirianiDataStructureSequence2007}.
A similar external memory adaptation was developed on the B-tree data
structure with the string B-tree of Ferragina and
Grossi~\cite{ferraginaStringBtreeNew1999}.
Both data structures support common trie operations such as prefix search,
predecessor/successor queries, insertions, and deletions in $\mathcal{O}(k/B
+ \log_B{n})$ I/Os, where $B$ is the block size.
Furthermore, Grossi and Italiano provide a framework for adapting
one-dimensional linked data structures to maintain multidimensional data by
storing two additional pointers and corresponding longest common prefix (LCP)
lengths to each node in order to generally
achieve search, insert, and delete operations in $\mathcal{O}(k + \log{n})$
time\footnote{The exact time complexity may vary depending on the underlying
data structure.}~\cite{10.1007/3-540-48523-6_34}.
In 2003, Crescenzi, Grossi, and Italiano showed that these techniques, when
applied to unbalanced binary search trees and AVL trees, were experimentally
competitive with existing string data structures~\cite{10.1007/3-540-44867-5_7}.

% Should we also introduce zip-trees here?
Recently there have been exciting developments in the area of randomized binary
search trees, with the introduction of the zip-tree data structure by
Tarjan, Levy, and Timmel~\cite{tarjanZipTrees2021a}, which is a simple,
practical, and efficient binary search tree algorithm that only requires
$\mathcal{O}(\log{\log{n}})$ bits of metadata per node w.h.p.
Several simple variants called zip-zip trees were introduced in 2023
showing how to improve its depth while maintaining a low metadata cost,
supporting strong history independence, biased keys, and/or partially
persistent~\cite{gilaZipzipTreesMaking2023}.
The zip-tree and its variants are notably incompatible with the existing
techniques for adapting one-dimensional data structures to support strings,
since they use at most $\mathcal{O}(\log{\log{n}})$ bits of metadata per node which
would be overshadowed by the $\mathcal{O}(\log{k})$ bits required to store LCP
length metadata for strings of moderate lengths, and since they
rely on efficient and simple `zip' and `unzip' procedures for
update operations as opposed to rotations.

\subsection*{Our Results.}
In this paper, we introduce several simple variants of zip-zip trees to support
high-dimensional variable-length data that is sorted lexically, such as
character strings.
We first introduce the \emph{zip-trie}, which supports such data while using
half as many pointers and exponentially less metadata per node than would be
required by existing methods.
The zip-trie is a simple modification of the zip-zip tree which
uses word-level parallelism to perform common trie operations such as prefix
search, range query, predecessor/successor queries, and insertions and deletions
efficiently.
Despite using far less memory than existing methods, the zip-trie is still
\textit{LCP-aware}, meaning that its time depends
on the length of the longest common prefix (LCP) between the query and the
strings in the data structure rather than the entire query length.
This property is especially useful when inserting
or searching for long strings that only share a short prefix with the strings in
the data structure.
For example, a recent collection of biosynthetic gene clusters, ABC-HuMi
\cite{hirschABCHuMiAtlasBiosynthetic2024}, contains around 19,000 genes whose
median length ($k$) is over 18,000 nucleotides, but whose median LCP length is
only 7.
The time for basic operations on this structure, for LCP length $\ell$, is
$\mathcal{O}(\ell/\alpha+\log n)$~(\Cref{theorem:zip-trie}).
We also introduce the \emph{parallel zip-trie} which achieves good theoretical
parallel bounds, both in the \emph{practical PRAM} model---a parallel extension
of the practical RAM model---and in the \emph{parallel external memory} (PEM)
model~\cite{argeFundamentalParallelAlgorithms2008}, where processors are allowed
to perform I/Os in parallel.

Both the memory improvements and the parallelization techniques we introduce
are general and can be applied to the existing framework of Grossi and Italiano
to adapt any existing one-dimensional linked data structures to maintain
multidimensional data.
We show how applying these techniques to the string B-tree of Ferragina and
Grossi~\cite{ferraginaStringBtreeNew1999} creates a span and work optimal
string dictionary data structure in the parallel external memory model, with
$\mathcal{O}(\log_B{n})$ I/O span and with $\mathcal{O}(\frac{k}{\alpha B} +
\log_B{n})$ I/O work, and also yields good results in the practical PRAM model.
These results, stated explicitly in
\Cref{thm:string-b-tree-pem,thm:string-b-tree-pram}, are to the best of our
knowledge the best known bounds for parallel trie operations in the PRAM and
PEM models.

Finally, we performed experiments comparing both our sequential and parallel
zip-trie data structures to the best previous sequential dynamic trie library,
\texttt{c-trie++} \cite{tsurutaCtrieDynamicTrie2022}, on real world datasets.
% We did apply to the skip-list and we have data... but it felt less relevant
% We also applied our parallelization techniques to the skip-list data structure,
% which we believe may now be among the simplest applications of parallel tries.
% ABC-HuMi, and maybe virus genomes if we have time
We found that our sequential zip-trie data structure, despite being
significantly simpler, outperforms the \texttt{c-trie++} in terms of trie
construction and search operations when the strings are large, and performed
competitively in other settings.
\texttt{c-trie++} does not currently provide predecessor/successor queries, prefix
search, or range queries, so we were unable to compare these operations.
Furthermore, unlike our trie variants, \texttt{c-trie++} is not easily
parallelizable.

\subsection*{Related Work.}
% The most closely related data structures to the ones we present in this paper
% are the string B-trees by Ferragina and
% Grossi~\cite{ferraginaStringBtreeNew1999}, which we discuss in more detail 
% in \cref{sec:parallel}, and the self-adjusting skip list data structure of Ciriani et
% al.~\cite{cirianiDataStructureSequence2007}.
% Although
% these data structures have been analyzed in
% the (sequential) external memory model, to the best of our knowledge,
% neither structure has been adapted for
% the PRAM and PEM settings (as we have done in this paper for zip-tries).
%
Not surprisingly, there is a wealth of additional related work, including
work on string data structures, parallel algorithms (e.g., for the PRAM and
PEM models), and external-memory algorithms.
Rather than review this material here, we include a review
of other related work in an appendix.

%% file: paradigm.tex
\section{The Bookend Search Paradigm}\label{sec:paradigm}

In this section we discuss techniques for searching for multi-dimensional
keys, such as character strings.

\subsection*{A Review of Bookend Searching.}
We begin by reviewing
a general paradigm introduced by Grossi and Italiano
for maintaining multi-dimensional keys in a data structure that would otherwise
only support one-dimensional keys~\cite{10.1007/3-540-48523-6_34}.
Naively, a data structure containing $n$ one-dimensional keys that supports access
operations in $A(n)$ time can be adapted to support $k$-dimensional keys with
access operations in $\mathcal{O}(A(n) \times k)$ time, since each individual
comparison may take up to $\mathcal{O}(k)$ time.
Grossi and Italiano realized that by exploiting the lexical ordering of the
keys, the access time could be reduced to $\mathcal{O}(A(n) + k)$.
The fundamental idea is simple---avoid unnecessary comparisons and reuse the
results of previous comparisons. 
Specifically, define a set of nodes $\pi_v$ as the set of nodes that are
traversed when searching for a key $v$, i.e., $v$'s ancestors, and let $\pi_v^-$
and $\pi_v^+$ be $v$'s immediate lexical predecessor and successor in $\pi_v$, or
$-\infty$ and $\infty$ if $v$ has no predecessor or successor, respectively.
Grossi and Italiano showed that by having each node $v$ store a pointer to
$\pi_v^-$ and $\pi_v^+$, and by storing the lengths of the longest common prefix
(LCP) between $v$ and $\pi_v^-$ and $\pi_v^+$, 
which we refer to as $\ell_{v, \pi_v^-}$ and $\ell_{v, \pi_v^+}$, respectively, unnecessary comparisons can be
avoided.
They provide a procedure called \textsc{k-Access}, which uses this metadata to
efficiently search for a search key $x$ in the data
structure~\cite{10.1007/3-540-48523-6_34}.
We call this technique \emph{bookend} searching, since it maintains 
predecessor and successor information.

\subsection*{Our \textsc{k-Compare} Bookend Search Method.}
We provide a simplified procedure we call \textsc{k-Compare} that can be
invoked whenever a comparison is needed to support bookend search,
but without each node explicitly needing to store any pointers to its ancestors.
See \Cref{fig:k-comp}.
To understand how this procedure works, consider one step in a search for a
key $x$ in a linked data structure, like a binary search tree.
Let $v$ be the node currently being compared to $x$, and assume that $x$ shares
its longest common prefix with its successor $\pi^+$ out of the nodes in $\pi =
\pi_v$.
Let's consider the following cases:
\begin{enumerate}[label=(\alph*)]
	\item If $\ell_{x, \pi^+} > \ell_{v, \pi^+}$, then $v$ differs from $\pi^+$
	earlier than $x$ does, and since $v \prec \pi^+$ ($v$ precedes $\pi^+$ lexically), it must also be true that
	$v \prec x$. Thus, the result of this comparison is $+$ with LCP length 
	$\ell_{v, \pi^+}$.

	\item If $\ell_{x, \pi^+} < \ell_{v, \pi^+}$, then $x$ differs from $\pi^+$
	earlier than $v$ does, and since $x \prec \pi^+$, it must also be true that
	$x \prec v$. Thus, the result of this comparison is $-$ with LCP length
	$\ell_{x, \pi^+}$.

	\item If $\ell_{x, \pi^+} = \ell_{v, \pi^+}$, then an actual comparison must
	be performed to not only determine the LCP length between $x$ and $v$, but
	also to determine the result of the comparison.
\end{enumerate}
We obtain not only a comparison result used for traversal, but also 
an LCP length which can be used to update $\pm$ and $\ell_{x, \pi^\pm}$ during
the course of the traversal that can be used during update operations.
The only expensive multi-dimensional LCP length computations performed start from
the result of the previous such computation, resulting in at most
$\mathcal{O}(\ell)$ additional time spent overall, where $\ell$ is the length of
the longest common prefix between the search key $x$ and the keys in the data
structure, and where $\ell$ is at most $k$.
This dependence on $\ell$ as opposed to $k$ is why we refer to this algorithm as
being \textit{LCP-aware}.

% \ofek{I feel like this is the most succinct procedure we can hope to achieve. I think this is a neat simplification/improvement over the algorithm presented in the original paper. I am not sure the best way to state this, giving both the original authors and us appropriate credit.}

\begin{figure}[hbtp]
    \centering
	\begin{algorithmic}[]
		\Procedure{k-Compare}{$x$, $v$, $\pm$, $\ell_{x, \pi^\pm}$}
			\State \IF $\ell_{x, \pi^\pm} > \ell_{v, \pi^\pm}$ \THEN \Return
			($\pm$, $\ell_{v, \pi^\pm}$)
			\State \IF $\ell_{x, \pi^\pm} < \ell_{v, \pi^\pm}$ \THEN \Return
			($\mp$, $\ell_{x, \pi^\pm}$)
			\Statex
			\State $\ell_{x, v} \gets \ell_{x, \pi^\pm} + $ \textsc{LCP}($x[\ell_{x,
			\pi^\pm}:k]$, $v[\ell_{x, \pi^\pm}:k]$)
			\State \Return (\textsc{Compare}($w[\ell_{x, v}]$, $v[\ell_{x, v}]$),
			$\ell_{x, v}$)
		\EndProcedure
	\end{algorithmic}

    \caption{\label{fig:k-comp} A simple bookend search procedure that compares a key $x$ to
		another key $v$.
		The input variable $\pm$ is $-$ if $x$ shares a greater
		LCP with $\pi^-$ than with $\pi^+$, and is $+$ otherwise.
		This procedure returns both the comparison result and the length of the
		longest common prefix between $x$ and $v$.
		The \textsc{LCP} procedure computes the LCP length between two keys, while the \textsc{Compare} procedure compares
		two one-dimensional values, returning either $-$, $+$, or $=$.
	}
\end{figure}

\begin{fact}\label[fact]{fact:k-comp}
	A lexical comparison
	%  operation 
	between two multi-dimensional keys $x$ and $v$,
	sharing a longest common prefix of length $\ell_{x, v}$,
%  which would normally
	% take $\mathcal{O}(\ell_{x, v})$ time, can instead be done using the
	can be done using the
	\textsc{k-Compare} procedure
	% in $\mathcal{O}(1)$ time 
	without comparing characters
	if $x$ shares a longer
	prefix with $\pi^\pm$ than with $v$,
	and by comparing a substring of length $\ell_{x, v} - \ell_{x, \pi^\pm} (+1)$
	otherwise.
\end{fact}

Using \Cref{fact:k-comp}, we can 
achieve the following.
% succinctly 
% restate parts of Grossi
% and Italiano's results that are relevant to our work.

\begin{theorem}
	Let $\mathcal{D}_S$ be a linked data structure that 
        visits $A(n)$ elements
	during an access operation, where $n = |S|$.
	If each node $v \in S$ has unique immediate predecessor and successor
	ancestors $\pi_v^-$ and $\pi_v^+$ which must be traversed before reaching
	$v$ in an access operation, and if $v$ stores these LCP lengths,
	then the time spent on comparisons using the
	\textsc{k-Compare} procedure is a telescoping sum described in
	\Cref{fact:k-comp}, 
        which is at most $\mathcal{O}(\ell_{x, \pi} + A(n))$,
	where $\ell_{x, \pi}$ is the length of the LCP between the search key $x$
	and the traversed keys $\pi$.
\end{theorem}

We note that the above result differs from that of Grossi and Italiano in that
it does not require storing pointers to $\pi_v^+$ and $\pi_v^-$ in each node.
As we discuss later, one of the appeals of zip-trees over previous
balancing schemes is a significant reduction in pointer manipulations and
metadata for the nodes of a search tree, which we carry over to our
zip-trie data structure.
Irving and Love~\cite{IRVING2003387} used a similar approach for suffix binary
search trees, and our provided procedure can be seen as a simple 
generalization of their approach.

\begin{corollary}\label{cor:seq-compare-string}
	If $\mathcal{D}_S$ is a collection of strings, then under the word RAM
	model it is possible to spend only $\mathcal{O}(\frac{\ell_{x, v} - \ell_{x,
	\pi^\pm}}{\alpha})$ time per \textsc{k-Compare} operation, where $\alpha$ is
	the number of characters that can be packed into a single machine word;
	hence, it is possible to spend only $\mathcal{O}(\frac{\ell_{x,
	\pi}}{\alpha} + A(n))$ time overall on comparison operations.
\end{corollary}

This result follows directly from efficient string comparisons we discuss
in~\Cref{sec:seq-string-comparisons}.

%% file: zip-tree.tex
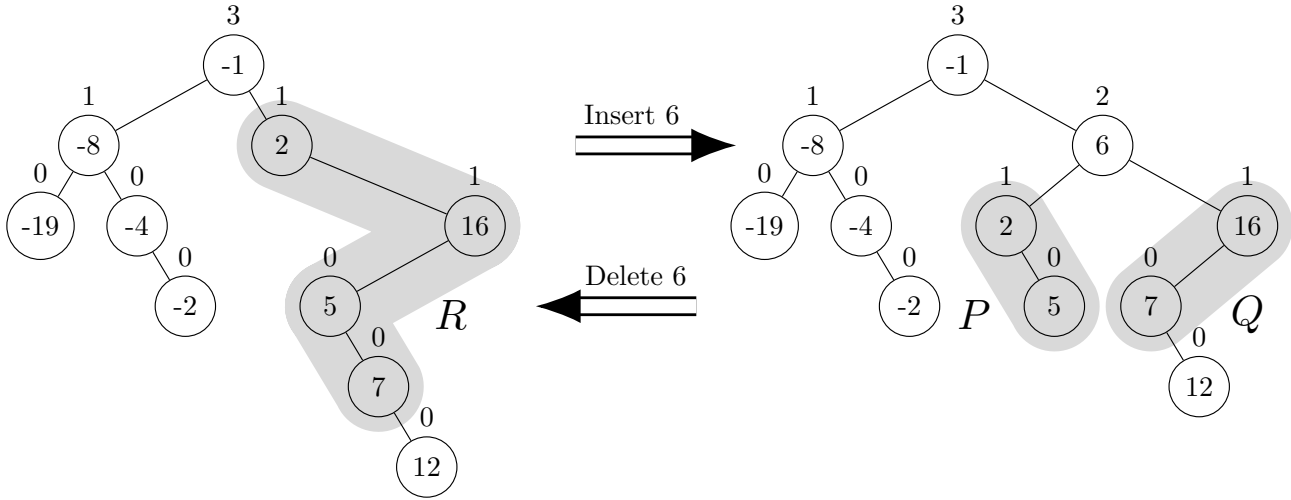
\begin{figure*}[bt]
    \centering
    \resizebox{\linewidth}{!}{\input{tikz/zip.tikz}}
    
    \caption{\label{fig:zip}
		% How insertion in a zip tree is done via unzipping
		% and  deletion is done via zipping.
		How insertion of a new node $x$ into a zip tree is done via `unzipping'
		and deletion of a node $x$ is done via `zipping'.
		$P$ and $Q$ refer to all updated nodes with smaller and larger keys than
		$x$, respectively, while $R$ refers to the union of these nodes.
	}
\end{figure*}

\section{A Review of Zip Trees and their Variants}

In this section, we review the zip tree data structure introduced by Tarjan,
Levy, and Timmel, and briefly review some of its variants.
Zip trees are isomorphic to skip lists but more efficient in both time and
space~\cite{tarjanZipTrees2021a}.
During the course of a search, both data structures visit the same
keys, except the skip list may visit the same key multiple times on
different levels.
Furthermore,
zip trees perform update operations using `zipping' and `unzipping' operations,
which, as opposed to common tree rotations,
% because the zip tree's use of zipping and unzipping operations
allow for more efficient pointer manipulation and reduced metadata
storage.
See \Cref{fig:zip}.
In terms of space usage,
common implementations of skip lists contain on average two copies
of every node in the list, while zip trees contain only one.\footnote{Other
implementations of skip lists exist, such as those that use variable-sized nodes,
but those also require at least as much and often far more space than zip trees.}
With regards to different balanced binary search tree implementations, zip trees
require only $\mathcal{O}(\log{\log{n}})$ bits of metadata per node for rank
information that is used to maintain the tree's balance.
More recently, zip-zip trees have been introduced as a simple variant of zip
trees that change the zip trees' ranks in order to improve the expected
search depth of nodes in the tree to $1.39 \log{n}$ from $1.5 \log{n}$, along
with adding many more desirable qualities~\cite{gilaZipzipTreesMaking2023}.

In a \emph{zip tree}, each node picks a random `rank' $r_1$, uniformly from a geometric
distribution with success probability $1/2$.
Nodes in the tree are max-heap ordered by rank, with ties broken in favor of
smaller keys.
Insertion and deletion in a zip tree are done by simple ``unzip'' and ``zip''
operations, respectively, which preserve this max-heap property.
To insert a new node $x$ into a zip tree, we first search for $x$ until reaching
the node $y$ that $x$ will replace, i.e., the first node with rank less than or
equal to that of $x$ (with a strict inequality if $y$'s key is less than $x$'s
key).
We then ``unzip'' the rest of $x$'s search path (including $y$) into two parts,
$P$, containing all nodes with smaller keys than $x$, and $Q$, containing all
nodes with larger keys (assuming that all keys are distinct).
We set the top nodes of $P$ and $Q$ to be the left and right children of $x$,
respectively, and let $x$ replace $y$ as the child of $y$'s parent.
Note that $P$ and $Q$ become the right (resp. left) spines of $x$'s left (resp.
right) subtrees. 
To delete a node $x$, we first find $x$ and then ``zip'' the spines of its two
children, $P$ and $Q$,
together into $R$, merging them from top to bottom in order to preserve the max-heap
ordering.
The root of $R$ replaces $x$ in the tree.
See \Cref{fig:zip}.
The left (resp. right) spine of a node is the path from that node to its
leftmost (resp. rightmost) descendant.
Pseudo-code for these operations can be found in~\cite{tarjanZipTrees2021a,gilaZipzipTreesMaking2023}.

Zip trees, as described above, experience sufficient rank collisions as to
unbalance the tree, hurting the expected search depth of nodes.
To address this, \emph{zip-zip trees} were introduced, adding onto each rank a
second value $r_2$, drawn independently from a uniform distribution, that is
used to break ties~\cite{gilaZipzipTreesMaking2023}.
Gila, Michael, and Tarjan showed that it is sufficient to pick $r_2$ from a
relatively small interval of size $\mathcal{O}(\log^c{n})$ for a small constant
$c$ to achieve the desired improvement in expected search depth, requiring only
$\mathcal{O}(\log{\log{n}})$ bits to store the entire rank information.
They also showed how to make zip trees \emph{external}, where items are stored only in external nodes of the tree while
the internal nodes only contain keys, \emph{partially persistent}, where queries
can be made on any version of the tree including past versions, and
\emph{biased}, where each key is favored (placed closer to the root) base on an
associated weight~\cite{gilaZipzipTreesMaking2023}.

%% file: tikz/zip.tikz
\begin{tikzpicture}[n/.style = {circle, draw, minimum width = 0.75cm, minimum height=0.5cm},a/.style = {line width=1pt, double distance=5pt,
	arrows = {-Latex[length=0pt 4 .5]}}]
	\def\ziptreeone{
		{-19/0/2},
		{-8/1/1},
		{-4/0/2},
		{-2/0/3},
		{-1/3/0},
		{2/1/1},
		{5/0/3},
		{7/0/4},
		{12/0/5},
		{16/1/2}%
	}

	% Draw all the nodes
	\foreach \key\rank\height [count=\i] in \ziptreeone{
		\pgfmathsetmacro\x{0.6*\i}
		\node[n,label=\rank] (\key) at (\x, -\height) {\key};
	}

	\draw[-] (-1) -- (-8) -- (-19)
		(-8) -- (-4) -- (-2)
		(12) -- (7) -- (5) -- (16) -- (2) -- (-1);

	\node[scale=1.5] (R) at (5.7, -3.1) {$R$};

	\begin{pgfonlayer}{background}
		\fill[gray,opacity=0.3] \convexpath{2,16,5,7,5,16}{16pt};
	\end{pgfonlayer}

	\draw [a] (7.25,-1) -- node[label={[xshift=-0.3cm]Insert 6}] {} (9.25,-1);
	\draw [a] (8.75,-3) -- node[label={[xshift=0.25cm]Delete 6}] {} (6.75,-3);

	\def\ziptreetwo{
		{-19/0/2},
		{-8/1/1},
		{-4/0/2},
		{-2/0/3},
		{-1/3/0},
		{2/1/2},
		{5/0/3},
		{6/2/1},
		{7/0/3},
		{12/0/4},
		{16/1/2}%
	}

	% Draw all the nodes
	\foreach \key\rank\height [count=\i] in \ziptreetwo{
		\pgfmathsetmacro\x{0.6*\i+9}
		\node[n,label=\rank] (\key) at (\x, -\height) {\key};
	}

	\draw[-] (6) -- (-1) -- (-8) -- (-19)
		(-8) -- (-4) -- (-2)
		(5) -- (2) -- (6) -- (16) -- (7) -- (12);

	\node[scale=1.5] (P) at (12.2, -3.1) {$P$};

	\node[scale=1.5] (Q) at (15.6, -3.1) {$Q$};

	\begin{pgfonlayer}{background}
		\fill[gray,opacity=0.3] \convexpath{2,5}{16pt};
		\fill[gray,opacity=0.3] \convexpath{16,7}{16pt};
	\end{pgfonlayer}
\end{tikzpicture}

%% file: zip-trie.tex
\section{Zip-Tries} \label{sec:zip-trie}

In this section we discuss several variants of the zip tree data structure
that we refer to as \emph{zip-tries},
which support high-dimensional keys such as strings.
We first discuss how the zip-trie data structure stores and updates its LCP
length metadata while performing zip and unzip operations.
We then discuss how the zip-trie exponentially reduces the memory overhead of
storing the LCP lengths by storing approximate values.
Finally, we introduce \emph{parallel zip-tries}.

\begin{theorem}\label[theorem]{theorem:zip-trie}
	Let $\ell$ be the length of the longest common prefix between a key, $x$, and
	the stored keys in a zip-trie, $T$. $T$ can perform prefix search,
	predecessor/successor queries, and insert/delete operations on $x$ in
	$\mathcal{O}(\frac{\ell}{\alpha} + \log{n})$ time, and using only
	$\mathcal{O}(\log{\log{n}} + \log{\log{\frac{\ell}{\alpha}}})$ bits of metadata per node
	w.h.p. in the word RAM model.\footnote{Our results also apply in the standard RAM model,
    with $\alpha=1$. We do not count children pointers
    as metadata, but these can be replaced 
    with $O(\log\log n)$-bit pointers~\cite{tiny}.}
\end{theorem}

\begin{corollary}
	Prefix search queries where one wishes to enumerate all $m$ keys sharing
	prefix $x$ can be done in $\mathcal{O}(m)$ additional time.
	Similarly, range queries can be done in $\mathcal{O}(\frac{\ell_1}{\alpha} +
	\frac{\ell_2}{\alpha} + \log{n} + m)$ time, where $\ell_1$ and $\ell_2$ are
	the LCP lengths of the range query bounds.
\end{corollary}

\subsection*{Update Operations.} \label{sec:zt-update}
In order for zip-tries to perform efficient string comparisons using the
\textsc{k-Compare} procedure described in \Cref{sec:paradigm}, zip-tries must
store and maintain accurate LCP length metadata on each of its nodes.
Specifically, each node, $v$, must store the LCP lengths between $v$ and its
immediate predecessor and successor ancestors, $\pi_v^-$ and $\pi_v^+$.

\begin{figure}[tbh]
	\centering
	% \resizebox*{\linewidth}{!}{\input{tikz/adjacent.tikz}}
	\resizebox*{!}{5cm}{\input{tikz/adjacent.tikz}}
	\caption{\label{fig:adjacent}
		A snapshot in the middle of a search operation for 
node, $x$, in an arbitrary binary search tree.
		Nodes along the path $\pi_v$ to the root are represented by circles.
		% Out of these nodes, both $w$ and $v$'s immediate successor $\pi_v^+$ is
		% $v$'s parent, while their immediate predecessor $\pi_v^-$ is the parent
		% of $v$'s spine (depicted in blue).
		The string keys of
		$x$ and $v$ must share their longest common prefix in $\pi_v$ with either $\pi_v^-$ or $\pi_v^+$.
		}
\end{figure}
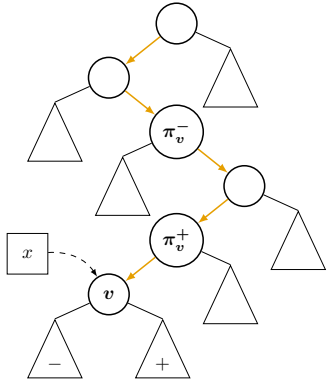

In this section, we show that the LCP length metadata in zip-tries can be
updated efficiently, increasing the total time by at most a constant factor,
during update operations.
To do this, we first show that no nodes outside the unzipping and zipping
paths, $P$ and $Q$, need to update their LCP length metadata, and we then show
that each of the nodes on the path can update their LCP length metadata in
constant time for either operation.

\begin{lemma}\label{lemma:zt-unchanged}
	Any node, $w$, outside the unzipping and zipping paths, $P$ and $Q$, or the
	target node, $x$, maintains the same immediate predecessor and successor
	ancestors, $\pi_w^-$ and $\pi_w^+$, after an update operation.
	Consequently, $w$'s LCP length metadata does not need to be updated.
	See \Cref{fig:zt-update}. 
\end{lemma}

\begin{proof}
	Consider an insertion operation, where $P$ and $Q$ refer to nodes along the
	unzipping path that are smaller and larger than the inserted node, $x$,
	respectively.
	After the insertion, $P$ and $Q$ become the right (resp. left) spines of $x$'s
	left (resp. right) subtrees,
	maintaining their same relative order, and no other nodes will
	move.
	Let us consider some arbitrary node, $w$, outside the unzipping and zipping
	paths, $P$ and $Q$, and assume that one of $w$'s immediate predecessor or
	successor ancestors changes.
	For $w$'s ancestors to change at all, they must be a
	descendant of $x$.
	Let's assume w.l.o.g. that $w \prec x$, i.e., that $w$ exists on the left
	subtree of a node, $p$ in $P$ both before and after the insertion.
	The only change to $w$'s ancestors is the addition of $x$ and the removal of
	all the nodes in $Q$.
	None of these nodes precede $w$, and consequently $w$'s immediate
	predecessor ancestor does not change.
	The only way for $w$'s immediate successor ancestor to change is if $w$'s
	immediate successor was part of $Q$ before the insertion or becomes $x$
	after the insertion.
	However, recall that one of $w$'s ancestors is $p$ both before and after the
	insertion, where $p \succ w$, and $p$ precedes not only $x$ but also all
	nodes in $Q$%
	, a contradiction.
	% .
	% Therefore, it is also not possible for $w$'s immediate successor ancestor to
	% change.
	Since zip trees are strongly history independent, deletion operations must share this
	property.
	% Due to symmetry, deletion operations must have the same property (this is
	% also apparent due to zip trees' history independence).
\end{proof}

While it is common among balanced binary search trees for the
only modified nodes during an insertion operation to be those along the search path, e.g.,
\cite{rankbalancedtrees,adel1962algorithm,seidel1996randomized},
the results of \Cref{lemma:zt-unchanged} would not generally hold outside of
rank-balanced trees such as zip-trees.
% This is because either the parents of nodes off the search path, or parents of
% the spines that those nodes are on, may change during rotations.

\begin{figure*}
	\centering
	% \resizebox*{\linewidth}{!}{\input{tikz/adjacent.tikz}}
	\resizebox*{\linewidth}{!}{\input{tikz/zt-update.tikz}}
	\caption{\label{fig:zt-update}
		% How insertion in a zip-trie is done via `unzipping' and deletion is done
		% via `zipping'.
		% We refer to the paths $P$ highlighted in gray as the `unzipping' and
		% `zipping' paths, respectively.
%		\ofek{In removal this is not really a `path' since it is disconnected. Probably need to reword this in a few places.}
		Update operations in a zip-trie, showing the LCP length metadata between
		each node and their immediate predecessor and successor ancestors on the
		left and right of each node and in blue and red, respectively.
		Note how the LCP metadata of nodes outside of $P$, $Q$, and the target
		node, $x$, are not updated, while only the successor (respectively predecessor)
		LCP lengths of nodes in $P$ (respectively $Q$) are updated.
		The nodes' ranks are not shown here for clarity,
		see \Cref{fig:zip} for a similar diagram in a zip tree showing the ranks.
	}
	\vspace*{-\medskipamount}
\end{figure*}
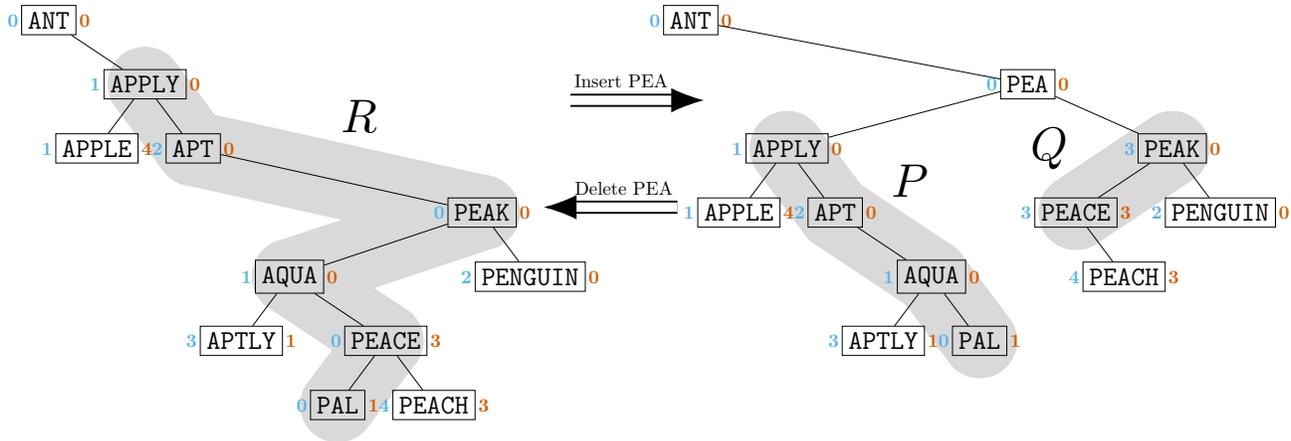

The only remaining nodes that may need to update their LCP length metadata are
nodes along the unzipping and zipping paths, $P$ and $Q$, and the target node
$x$.
We will discuss how to update the LCP length metadata of one
such node, $v$, in constant time for either operation.
% , and we will assume
% w.l.o.g. that $v < x$, i.e., that $v$ is in $P$.
Before each operation, $v$ knows the LCP length between itself and its
ancestors, $\pi_v^-$ and $\pi_v^+$, while during the course of each operation,
$v$ must update the LCP length metadata to reflect LCP lengths between its new
ancestors, ${\pi'_v}^-$ and ${\pi'_v}^+$.
% Note that even though nodes do not store pointers do their ancestors, this
% information can be determined during the course of update operations in constant
% time since all ancestors $\pi_v$ of $v$ must be visited during the course of the
% operation.

We first consider the case of insertion.
Let's start by assuming that $v$ is part of $P$, i.e., $v \prec x$.
Consider the set of $v$'s ancestors, $\pi_v$, during the course of the
insertion.
All nodes in $\pi_v$ that precede $v$ must either be part of $P$, in which case
they remain ancestors of $v$ on the right spine of $x$ left subtree, or they must also be
ancestors of $x$ after the operation, in which case they remain entirely
unchanged.
In both cases, $v$'s immediate predecessor ancestor, $\pi_v^-$, remains
unchanged and consequently there is nothing to do for $\ell_{v, \pi_v^-}$.
Regarding $v$'s immediate successor ancestor, regardless of what it was before,
must become $x$.
Recall that during the course of the insertion, all nodes in $P$ are compared
against $x$ using the \textsc{k-Compare} procedure, which also returns the LCP
length shared between the two keys.
Consequently, $v$'s LCP length metadata with its successor, $x$, is already
known, i.e., $\ell_{v, \pi_v^+} = \ell_{v, x}$.
Due to symmetry, if $v$ is part of $Q$, i.e., $v \succ x$, then we update
$\ell_{v, \pi_v^-} = \ell_{v, x}$.

Deletion is the inverse operation.
Let's again assume that $v$ is part of $P$.
Due to symmetry, $v$'s preceding ancestors do not change, and consequently
we must only update $v$'s LCP length with its immediate successor ancestor,
${\pi'_v}^+$.
Consider the possible candidates of ${\pi'_v}^+$.
Either this node is a member of set $Q$, or it is the immediate successor of
$x$.
In both cases, the LCP length between $x$ and ${\pi'_v}^+$ is known using either
$x$'s or ${\pi'_v}^+$'s successor or predecessor metadata, respectively.
Using the triangle inequality, we can update $\ell_{v, {\pi'_v}^+} =
\min(\ell_{v, x}, \ell_{x, {\pi'_v}^+})$.
Due to symmetry, if $v$ is part of $Q$, then we update $\ell_{v, {\pi'_v}^-} =
\min(\ell_{v, x}, \ell_{x, {\pi'_v}^-})$.

\subsection*{Reducing LCP Metadata Size.}
In the previous section, we showed how to efficiently update the LCP length
metadata associated with each node during zip-trie update operations.
In the case of collections of short strings, or strings that share very short
prefixes with each other---i.e., $\ell = \mathcal{O}(\log{n})$---the LCP length
metadata does not contribute significantly to the memory overhead of the
zip-trie, and storing the exact LCP lengths is possible.
We refer to this as the \emph{memory-intensive zip-trie}, (MI-ZT), and we
evaluate its performance experimentally.
However, in the case of longer strings, in order to preserve the zip-trie's
low metadata overhead, we store approximate LCP lengths using exponentially
fewer bits.
These approximations will always be lower bounds on the true LCP lengths, and
since these approximations will be close enough to the true values, we 
prove that we only increase the time spent on string comparisons by a constant
factor.

To do this, zip-tries only represent approximate LCP lengths defined by a
paradigm detailed below.
Let $A(n)$ be the number of times that \textsc{k-Compare} is invoked during the
course of an operation, and let $A(n) = \Theta(f(n))$ for some function $f(n)$.
For zip-tries, we pick $f(n) = \log{n}$.
Instead of storing the exact LCP length $\ell$ between a node, $v$, and its
ancestor, we instead store an approximate LCP length of the form $2^a
\times b$, where $a$ and $b$ are non-negative integers with maximum values
$\log{\ell}$ and $f(n)$, respectively.
For any exact LCP length $\ell$, we can set $a = \lfloor \log{\frac{\ell}{f(n)}}
\rfloor$ and $b = \lfloor \frac{\ell}{2^a} \rfloor$, i.e., the largest
approximate LCP length that is less than or equal to $\ell$.

\begin{lemma}
	Representing $\ell$ as described above requires at most $\log{\log{\ell}} +
	\log(2 f(n))$ bits.
\end{lemma}

\begin{proof}
	The largest possible value of $a$ is $\log{\ell}$, which requires only
	$\log{\log{\ell}}$ bits to represent.
	We can lower bound $a$ by $\log{\frac{\ell}{f(n)}} - 1$, and consequently
	$2^a$ by $\frac{1}{2} \frac{\ell}{f(n)}$.
	Using this lower bound, we see that $b$ can be at most $2 f(n)$, requiring
	only $\log(2 f(n))$ bits to represent.
\end{proof}

One of the main advantages of trie data structures using \textsc{k-Compare} is
that the time taken to perform comparisons during the course of an operation is
proportional to the LCP length between the target key $x$ and the keys in the
trie.
More specifically, the time spent on comparisons is $\mathcal{O}(\ell +
A(n))$.
In order to maintain this property, we must ensure that the approximations are
close enough to the true LCP lengths.

\begin{lemma} \label{lemma:approx-lcp-diff}
	Let $\ell$ and $\tilde \ell$ represent the true and approximate LCP lengths,
	respectively, between two keys.
	The maximum difference between $\ell$ and $\tilde \ell$ is at most
	$\frac{\ell}{f(n)}$.
\end{lemma}

\begin{proof}
	We know that $b = \lfloor \frac{\ell}{2^a} \rfloor > \frac{\ell}{2^a} - 1$.
	Recalling that $\tilde \ell = 2^a \times b$, we can substitute our lower
	bound for $b$ to obtain that $\tilde \ell > \ell - 2^a$.
	Rearranging this equation we see that $\ell - \tilde \ell < 2^a$.
	The proof follows directly from the fact that trivially $2^a \leq \frac{\ell}{f(n)}$.
\end{proof}

We are now equipped to show that the use of approximate LCP lengths does not
increase the asymptotic time spent on comparisons.

\begin{lemma} \label{lemma:approx-lcp-time}
	A data structure, $\mathcal{D}_S$, spending $\mathcal{O}(\ell + A(n))$ time on comparisons
	using exact LCP lengths will spend at most $\mathcal{O}(\ell)$ additional
	time using approximate LCP lengths.
\end{lemma}

\begin{proof}
	By the discussion above, the maximum difference between the true
	and approximate LCP lengths is at most $\frac{\ell}{f(n)}$, resulting in at
	most an additional $\frac{\ell}{f(n)}$ time spent per comparison.
	Since $A(n)$ comparisons are made during the course of the operation, and
	since $A(n) = \Theta(f(n))$, the total additional time spent on comparisons
	is at most $\mathcal{O}(A(n) \times \frac{\ell}{f(n)}) = \mathcal{O}(\ell)$.
\end{proof}

\begin{corollary}
	If $\mathcal{D}_S$ is a collection of strings, then under the word-RAM
	model, $\mathcal{D}_S$ will spend at most
	$\mathcal{O}(\frac{\ell}{\alpha})$ additional time using approximate LCP
	lengths.
\end{corollary}

Since over-estimating $A(n)$ may only improve the approximation, we assume that
$A(n)$ is a constant throughout the lifecycle of the trie, specifically the
largest value of $A(n)$ that will be encountered, i.e., $\mathcal{O}(\log{n})$
w.h.p. when $n$ is the maximum number of nodes in the trie.
\Cref{theorem:zip-trie} follows from the above results along with the results
of the previous section.

\subsection*{Parallel Zip-Tries.} \label{sec:parallel-zt}
In this subsection, 
we introduce \emph{parallel zip-tries}, which are a variant of zip-tries
suited for high-dimensional keys.
We present these results as a general framework for parallelizing subsequent
calls to \textsc{k-Compare} during the course of an operation.
We note that the \textsc{k-Compare} algorithm described in
\Cref{sec:paradigm} runs in constant time other than the time spent computing
the LCP length, which may take up to $\ell_{x, v} - \ell_{x, \pi^\pm}$ time.
We assume that we have a parallel LCP oracle \textsc{Parallel-LCP} that can
compute the LCP length between two $k$-dimensional keys using $\mathcal{O}(k)$
work and $\mathcal{O}(1)$ span.
% We show such an oracle for strings in an appendix.
We show such an oracle for strings in \Cref{sec:pram-string-comparisons}.
Consider naively substituting \textsc{LCP} for \textsc{Parallel-LCP} in the
\textsc{k-Compare} algorithm.
If $\ell_{x, \pi^\pm}$ were to become large early on, i.e., one of the first
nodes in $\pi$ share a long prefix with $x$, then all subsequent comparisons
would be very cheap and we may spend at most $\mathcal{O}(k)$ work during the
course of the operation.
However, if $\ell_{x, \pi^\pm}$ were to stay small, we may spend up to
$\mathcal{O}(k)$ work \textit{per comparison}, resulting in $\mathcal{O}(k
\times A(n))$ work during the course of the operation, where $A(n)$ is the
number of times that \textsc{k-Compare} is invoked.
This result makes no use of \textsc{k-Compare}'s ability to reuse comparisons.
We have to be a bit more clever in order to achieve optimal work.

\begin{theorem}\label[theorem]{theorem:parallel-zt-k}
	If $\mathcal{D}_S$ invokes \textsc{k-Compare} $A(n)$ times during the course
	of an operation, then it can perform all comparison operations in
	$\mathcal{O}(A(n))$ span and $\mathcal{O}(k + A(n))$ work in CRCW PRAM,
	where $k$ is the dimensionality of the target key $x$.
\end{theorem}

\begin{proof}
	As we described in the above analysis, if we provide the entire remainder of
	$x$ and $v$ to \textsc{Parallel-LCP} at a time, then we may end up
	performing too much unnecessary work.
	On the other hand, if we provide too little, we may end up invoking
	\textsc{Parallel-LCP} too many times, hurting the span of our algorithm.
	The key observation is that these two conflicting goals can be balanced by
	providing only $\frac{k}{f(n)}$ items of $x$ and $v$ to
	\textsc{Parallel-LCP} at a time, for some function $f(n)$ where $A(n) =
	\Theta(f(n))$.
	We repeatedly invoke \textsc{Parallel-LCP} on these subsets of $x$ and $v$
	until we have computed the LCP length between the two keys.
	As before, it is still possible to perform up to $\mathcal{O}(k)$ work per
	comparison, but crucially this can only occur when $x$ and $v$ indeed share
	a long prefix.
	To analyze the work and span of our algorithm, we can consider the two
	possible results of a \textsc{Parallel-LCP} invocation:
	\begin{enumerate}
		\item The prefixes are different. This can occur at most once per
		comparison, and consequently $A(n)$ times in total.

		\item The prefixes are the same. While this can occur several times per
		comparison, due to the nature of the \textsc{k-Compare} algorithm, we
		will never compare these prefixes the same.
		Since each prefix is of size $\frac{k}{f(n)}$ and our key only has $k$
		items, this can only occur at most $f(n)$ times in total.
	\end{enumerate}

	Both cases only occur $\mathcal{O}(A(n))$ times overall,
	each time involving $\mathcal{O}(\frac{k}{A(n)})$ work.
\end{proof}

Note that while the above result is optimal when $x$ is already in the data
structure, i.e., when $\ell = k$, it performs poorly when $x$ shares only a very
short prefix with the keys in the data structure, i.e., when $\ell \ll k$.
There are several natural ways to try to reduce the work considerably in this
case.
One way is, for each \textsc{Parallel-LCP} invocation, providing subsets of only
one item at a time, and then only if they are the same, providing two, then
four, and so on, continuously doubling until the prefixes are different.
This way, we end up comparing at most $2(\ell_{x, v} - \ell_{x, \pi^\pm})$ extra
items per comparison, or $2k$ extra items overall.
While such an algorithm only spends the desired $\mathcal{O}(k +
A(n))$ work spent on comparisons, since it may double up to $\log{\ell}$ times
per comparison, it may spend
% Overall, such an algorithm would compare at most $2k$ extra items, resulting in
% the desired $\mathcal{O}(k + A(n))$ work spent on comparisons.
% However, we may double up to $\log{\ell}$ times per comparison, resulting in
$\mathcal{O}(\log{\ell} \times A(n))$ span overall.

Alternatively, in order to reduce span, we may avoid resetting the number of
items we compare between different comparisons, opting to start where we left
off in the previous comparison, incurring the $\log{\ell}$ doubling cost only
once in total as opposed to once per comparison.
However, the amount of extra work incurred would no longer be a
telescoping sum, and could be up to $\mathcal{O}(k)$ \textit{per comparison},
resulting in $\mathcal{O}(k \times A(n))$ work overall.

\begin{theorem} \label[theorem]{theorem:parallel-zt-l}
	Let $\ell$ be the longest common prefix length between $x$ and the
	keys in $\mathcal{D}_S$.
	If $\mathcal{D}_S$ invokes \textsc{k-Compare} $A(n)$ times during the course
	of an operation, then it can perform all comparison operations in
	$\mathcal{O}(\log{\ell} + A(n))$ span and $\mathcal{O}(\ell + A(n))$ work in
	CRCW PRAM.
\end{theorem}

\begin{proof}
	We face a similar dilemma as in the proof of \Cref{theorem:parallel-zt-k}.
	If we start our doubling procedure by comparing too few items we may end up
	hurting our span, while if we start by comparing too many, we may end up
	hurting our work.
	In this case, we can balance these two goals by starting our
	\textsc{Parallel-LCP} procedure by comparing half as many items as we did
	before.
	The idea is that if we repeatedly compare too many items, our excess work
	keeps halving, overall incurring a constant factor of extra work.
	As before, we consider:
	\begin{enumerate}
		\item The prefixes are different.
		This can occur at most $\mathcal{O}(A(n))$ times overall.

		\item The prefixes are the same.
		Let $\delta$ be the current subset size provided to
		\textsc{Parallel-LCP}, and let $\delta$ start at 1.
		After $\log{\ell}$ invocations where the prefixes are the same, $\delta$
		will reach $\ell$, at which point no further `equal' invocations are
		possible.
		It is easy to see that the total number of `equal' invocations are at
		most $\log{\ell}$ plus the number of times where the prefixes differed,
		or $\log{\ell} + \mathcal{O}(A(n))$ times overall.
	\end{enumerate}

	Due to the nature of the \textsc{k-Compare} algorithm, we spend at most
	$\ell$ work on \textsc{Parallel-LCP} invocations where the prefixes are the
	same.
	And due to the nature of the halving procedure whenever we compare different
	prefixes, the total work incurred by such invocations is at most a constant
	multiple of the work incurred by the `equal' invocations.
\end{proof}

For strings, we can pack $\alpha$ characters per `item', and under the practical
PRAM model~\cite{eppsteinBriefAnnouncementUsing2017}, we can perform operations
on machine words in parallel in $\mathcal{O}(1)$ time.
Applying the above results to the parallel zip-trie data structure, we obtain:

\begin{theorem}
	Let $\ell$ be the length of the longest common prefix between a string key
	$x$ of length $k$, and the stored keys in a parallel zip-trie, $T$.
	$T$ can perform prefix search, predecessor/successor queries, and update
	operations on $x$ in $\mathcal{O}(\log{n})$ span and
	$\mathcal{O}(\frac{k}{\alpha} + \log{n})$ work under the practical PRAM
	model,
	or alternatively in
	% Alternatively, $T$ can perform the same operations in
	$\mathcal{O}(\log{\frac{\ell}{\alpha}} + \log{n})$ span and
	$\mathcal{O}(\frac{\ell}{\alpha} + \log{n})$ work.
\end{theorem}

%% file: tikz/adjacent.tikz
\begin{tikzpicture}[
    n/.style={draw, color=okabe1, circle, thick, minimum width = 0.75cm},
    t/.style={draw, color=okabe1, shape border rotate=90, isosceles triangle, isosceles triangle apex angle=50, thin, minimum width = 1cm},
    level 1/.style={sibling distance=2.5cm},
    level distance=1cm,
    % level 2/.style={sibling distance=4cm},
    % level 3/.style={sibling distance=7.5cm},
    % level 4/.style={sibling distance=2.5cm},
    te/.style={color=okabe1, child anchor=apex, thin},
    pe/.style={->,>=latex, color=okabe2, thick}
]
    \node[n] {}
        child { node[n] {}
            child { node[t, anchor=left side] {}
                edge from parent[te]
            }
            child { node[n] (u) {$\boldsymbol{\pi_v^-}$}
                child { node[t, anchor=left side] {}
                    edge from parent[te]
                }
                child { node[n] {}
                    child { node[n] (w) {$\boldsymbol{\pi_v^+}$}
                        child { node[n] (v) {$\boldsymbol{v}$}
                        % child { node[n, rectangle, minimum width=0.75cm, minimum height=0.75cm, thin] (u) {$\boldsymbol{u}$}
                            child { node[t, anchor=left side] {$-$}
                                edge from parent[te]
                            }
                            child { node[t, anchor=right side] {$+$}
                                edge from parent[te]
                            }
                            edge from parent[pe]
                        }
                        child { node[t, anchor=right side] {}
                            edge from parent[te]
                        }
                        edge from parent[pe]
                    }
                    child { node[t, anchor=right side] {}
                        edge from parent[te]
                    }
                    edge from parent[pe]
                }
                edge from parent[pe]
            }
            edge from parent[pe]
        }
        child { node[t, anchor=right side] {}
            edge from parent[te]
        }
    ;

    \node[n, rectangle, minimum width=0.75cm, minimum height=0.75cm, thin] (s) at ($(v) + (-1.5,0.75)$) {$x$};
    % \node[n, thin] (s) at ($(v) + (-1.5,0.75)$) {$\boldsymbol{s}$};
    \draw[->, >=latex, bend left, dashed, thin] (s) to (v);
\end{tikzpicture}

%% file: tikz/zt-update.tikz
\begin{tikzpicture}[
    n/.style = {
            rectangle, draw, minimum width = 0.75cm, minimum height=0.5cm,
            font=\ttfamily\Large
        },
        a/.style = {
            line width=1pt, double distance=5pt,
            arrows = {-Latex[length=0pt 2.5]}
        },
        label distance=-1pt,
        % label/.append style={font=\Large},
    ]\def\ziptree{
        {0/ANT/0/0},
		{2/APPLE/1/4},
        {1/APPLY/1/0},
        {2/APT/2/0},
        {5/APTLY/3/1},
        {4/AQUA/1/0},
        {6/PAL/0/1},
        % {4/PEA/(3)},
        {5/PEACE/0/3},
        {6/PEACH/4/3},
        {3/PEAK/0/0},
        {4/PENGUIN/2/0}%
        % {2/QUAIL/(0)},
        % {1/QUICK/},
        % {3/QUIZ/(3)},
        % {3/SEAL/}%
    }

    \foreach \height\word\prec\succ [count=\i] in \ziptree{
        % \pgfmathsetmacro\x{0.9*\i};
        \pgfmathsetmacro\y{-\height*1.2}
        \pgfmathsetmacro\x{0.9*\i}
        % \pgfmathsetmacro\y{-\height};
        % node with lcp label below
        \node[n, label={[font=\bfseries, okabe3, xshift=1pt]left:{\prec}}, label={[font=\bfseries, okabe7, xshift=-1pt]right:{\succ}}] (\word) at (\x, \y) {\word};   }

    \draw[-] (ANT) to (APPLY);
    
    \draw[-] (APPLY) to (APPLE); 
    \draw[-] (APPLY) to (APT);

    \draw[-] (APT) to (PEAK);

    \draw[-] (PEAK) to (AQUA);
    \draw[-] (PEAK) to (PENGUIN);
    
    \draw[-] (AQUA) to (APTLY);
    \draw[-] (AQUA) to (PEACE);

    \draw[-] (PEACE) to (PAL);
    \draw[-] (PEACE) to (PEACH);
    \node[scale=2.5] (R) at (6.7, -1.8) {$R$};

	\begin{pgfonlayer}{background}
        \fill[gray,opacity=0.3] \convexpath{APPLY,APT,PEAK,AQUA,PEACE,PAL,PEACE,AQUA,PEAK,APT}{20pt};
    \end{pgfonlayer}

	\draw [a] (10.65,-1.5) -- node[label={[xshift=-0.3cm]Insert PEA}] {} (13.15,-1.5);
    \draw [a] (12.65,-3.5) -- node[label={[xshift=0.25cm]Delete PEA}] {} (10.15,-3.5);

	\def\ziptree{
        {0/ANT/0/0},
		{3/APPLE/1/4},
        {2/APPLY/1/0},
        {3/APT/2/0},
        {5/APTLY/3/1},
        {4/AQUA/1/0},
        {5/PAL/0/1},
        {1/PEA/0/0},
        {3/PEACE/3/3},
        {4/PEACH/4/3},
        {2/PEAK/3/0},
        {3/PENGUIN/2/0}%
        % {2/QUAIL/(0)},
        % {1/QUICK/},
        % {3/QUIZ/(3)},
        % {3/SEAL/}%
    }

    \foreach \height\word\prec\succ [count=\i] in \ziptree{
        % \pgfmathsetmacro\x{0.9*\i};
        \pgfmathsetmacro\y{-\height*1.2}
        \pgfmathsetmacro\x{0.9*\i+12}
        % \pgfmathsetmacro\y{-\height};
		% node with lcp label below
        \node[n, label={[font=\bfseries, okabe3, xshift=1pt]left:{\prec}}, label={[font=\bfseries, okabe7, xshift=-1pt]right:{\succ}}] (\word) at (\x, \y) {\word};   }

    \draw[-] (ANT) to (PEA);

    \draw[-] (PEA) to (APPLY);
    \draw[-] (PEA) to (PEAK);

    \draw[-] (APPLY) to (APPLE);
    \draw[-] (APPLY) to (APT);

    \draw[-] (APT) to (AQUA);

    \draw[-] (AQUA) to (APTLY);
    \draw[-] (AQUA) to (PAL);

    \draw[-] (PEAK) to (PEACE);
    \draw[-] (PEAK) to (PENGUIN);

    \draw[-] (PEACE) to (PEACH);
	
    \node[scale=2.5] (P) at (17, -3) {$P$};
    \node[scale=2.5] (Q) at (19.6, -2.4) {$Q$};

	\begin{pgfonlayer}{background}
        \fill[gray,opacity=0.3] \convexpath{APPLY,APT,AQUA,PAL,AQUA,APT}{20pt};
        \fill[gray,opacity=0.3] \convexpath{PEAK,PEACE}{20pt};
    \end{pgfonlayer}
\end{tikzpicture}

%% file: string-b-tree.tex
\section{Parallel String B-Trees} \label{sec:string-b-tree}

\begin{figure*}[t!]
	\vspace*{-\bigskipamount}
    \centering
    \resizebox{\linewidth}{!}{\input{tikz/b-tree-branching.tikz}}
    \caption{\label{fig:b-tree-branching} An illustration of metadata stored
    in a string B-tree node.
%	Note that in practice this is an array and only pointers to the keys are
	% stored%
	% ,
	% rather the keys themselves%
	% .
%	These details were omitted for clarity.
}
\end{figure*}
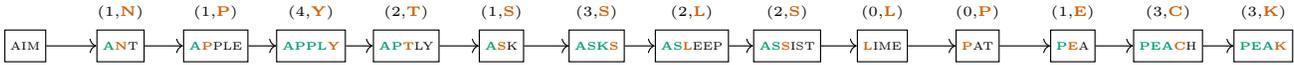

\begin{figure*}[b]
	\centering
	\resizebox{\linewidth}{!}{\input{tikz/b-tree-branching-pram.tikz}}
	% \input{tikz/b-tree-branching-pram.tikz}
	% \vspace{-0.5cm}
	\caption{\label{fig:b-tree-branching-pram} An illustration of the steps
	taken to determine the index $w$ of a key that shares the longest common
	prefix with the search key $x$ (\texttt{APPLICATION}).
	Refer to the text for an explanation of each step.
	In practice this is an array and only pointers to the keys are
	stored%
	,
	rather the keys themselves%
	.
	These details were omitted for clarity.}
\end{figure*}
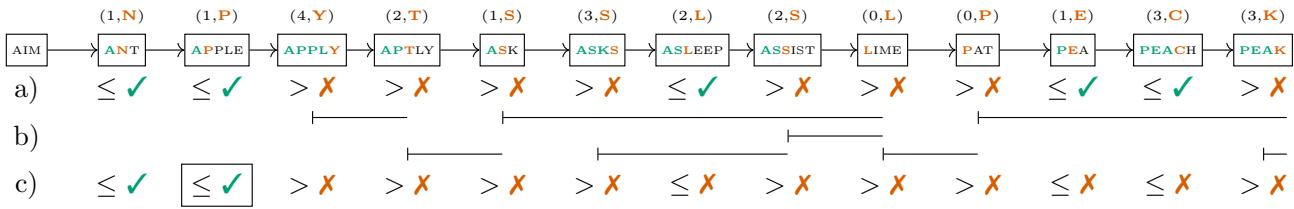

The string B-tree is a classical data structure for performing prefix search,
range queries, and similar trie operations on a collection of strings using an
optimal number of I/Os in the external memory model.
Specifically, the string B-tree performs the prefix search operation in
$\mathcal{O}(\frac{K}{B} + \log_B{n})$ I/Os, where $K$ is the number of machine
words required to store the search key, and $B$ is the number of machine words
that can be stored in a single block.
While its complexity, large memory usage, and overhead, particularly for update
operations, make a dynamic version of the string B-tree largely impractical in
practice~\cite{ferraginaStringAlgorithmsData2008}, we show how to apply the
parallelization framework we developed for the parallel zip-trie to
adapt the string B-tree to the practical PRAM model, pushing the bounds of span
while adding only a little increase in work.
Furthermore, we show how to maintain optimal I/Os while also achieving optimal
I/O span under the parallel external memory (PEM) model.

In string B-trees, each node of size $\mathcal{O}(B)$ is stored in
$\mathcal{O}(1)$ contiguous blocks of memory, and branches to $\mathcal{O}(B)$
children.
During operations, the string B-tree determines which child to follow by
performing only one comparison, for example using the \textsc{k-Compare}
procedure from \Cref{sec:paradigm}.
To know which of the $\mathcal{O}(B)$ keys in a node to compare against, the original string
B-tree used Patricia tries, while subsequent
papers~\cite{naSimpleImplementationString2004a,ferraginaStringAlgorithmsData2008,joo-youngEffectiveImplementationString2006,martinez-prietoPracticalCompressedString2016} 
have preferred to use a simpler algorithm, derived from the branching algorithm
for Bit-trees~\cite{fergusonBitTreeDataStructure1992}.

First, we describe relevant details of the string B-tree and the branching
algorithm.
Each internal node in a string B-tree stores pointers to $\mathcal{O}(B)$
children nodes, and in addition stores pointers to the first and last keys of
each child.
Using the simpler branching algorithm for string
B-trees~\cite{naSimpleImplementationString2004a}, each node stores these
pointers lexicographically in contiguous memory.
In addition to these pointers, each node also stores the LCP length shared
between two adjacent keys, and the first character where they differ (to avoid
expensive I/Os to look up these characters on the fly).
For simplicity, we assume that the metadata is stored in three separate arrays,
$k$, $\ell$, and $c$.
The $i$-th key has its pointer stored in $k_i$, the LCP length between itself
and the \textit{previous} key stored in $\ell_i$, and the first differing
character stored in $c_i$, where $c_i = k_i[\ell_i]$.
We depict this metadata in \Cref{fig:b-tree-branching}.

To achieve better performance in parallel models, each node stores an additional
array $n$ containing the index of the next key that has a smaller or equal LCP
value, or $\infty$ if no such key.
We now describe the three steps of the branching algorithm
from~\cite{naSimpleImplementationString2004a}:
\begin{enumerate}
	\item Determine a key at index $w$ that shares the longest common prefix
	with the search key $x$.

	\item Compare the search key $x$ with $k_w$ using \textsc{k-Compare},
	obtaining the LCP length.
	
	\item \label{step:ba-3} Find which key at index $y$ is either equal to $x$
	or is its immediate successor.
\end{enumerate}

We do not prove here why these steps correctly determine the correct index $y$
to follow in the string B-tree.
We refer the curious reader to view the full explanation as provided
in~\cite{naSimpleImplementationString2004a} for more details.

\begin{lemma} \label{lem:branching} The string B-tree branching algorithm (other
	than the call to \textsc{k-Compare}) can be implemented using $\mathcal{O}(1)$
	span and $\mathcal{O}(B \log{B})$ work in CRCW PRAM.
\end{lemma}

\begin{proof}
	Perhaps the trickiest part of this algorithm is the first step, which we
	break down into three substeps.
	These substeps are also depicted in \Cref{fig:b-tree-branching-pram}.

	\begin{enumerate}[label=\alph*)]
		\item \label{step:pb-a} First, each key at index $i$ concurrently
		compares itself against the search key $x$ at the $\ell_i$-th character
		of each string.
		Each key which compares $\leq$ the search key is deemed a `candidate'.

		\item \label{step:pb-b} Next, each key that is \textit{not} a candidate
		eliminates all candidates between itself and the next key that has a
		smaller LCP value.

		\item \label{step:pb-c} Finally, the index of the last remaining
		candidate is determined to be the desired index $w$.
		If there are none, set $w = 0$.
	\end{enumerate}

	\crefname{enumi}{step}{steps}
    \Crefname{enumi}{Step}{Steps}

	\Cref{step:pb-a} is trivially done by simultaneously comparing $c_i$ with
	$x[\ell_i]$ for each key $i > 0$ in parallel, taking constant span and
	$\mathcal{O}(B)$ work.

	We then perform the most complex step in this algorithm, \Cref{step:pb-b},
	where each key at index $u$ that is not a candidate eliminates all
	candidates between itself and the next key that has a smaller or equal LCP
	value, at index $n_u$.
	For convenience, we define this range as $[u, n_u - 1]$ so that the range is
	inclusive.
	Naively, this can be done by each non-candidate key $u$ simultaneously
	eliminating all candidates between itself and $n_u$, which may take up to
	$\mathcal{O}(B)$ work per key or $\mathcal{O}(B^2)$ work in total if there
	is significant overlap between the ranges.
	To reduce the work to $\mathcal{O}(B \log{B})$, we instead consider each key
	in the array as though it corresponds to a leaf in a perfect binary tree
	with $\mathcal{O}(\log{B})$ levels, which we refer to as a `range tree'.
	All nodes are initialized to 0.
	Since this is a perfect binary tree, we can compute the array position of
	any ancestor of a node in constant span and work.
	Similarly, we can compute the array position of the lowest common ancestor
	$a$ of each range in constant span and work.
	Consider the path $p$ from the start of the range to the end of the range,
	passing through $a$, denoting the first half of the range from the start to
	$a$ as $p_l$ and the second half from $a$ to the end as $p_r$.
	With $\mathcal{O}(\log{B})$ processors per range, we can in parallel set the
	right child of each node in $p_l$, the left child of each node in $p_r$, and
	the $n_u - 1$ node itself to 1 (using a bitwise OR operation).
	Since there are at most $\mathcal{O}(B)$ ranges, this step can be done using
	$\mathcal{O}(B \log{B})$ work and in constant span.
	A `1' in any node of the tree indicates that the corresponding key of any of
	its descendants (if it is an internal node), or of itself (if it is a leaf
	node), falls within one of these ranges, and is therefore eliminated from
	being a candidate.
	See \Cref{fig:range-tree}.
	%  for an illustration of this process.

	\begin{figure}[ht!]
        \vspace*{-\medskipamount}
        \centering
        % \resizebox{\linewidth}{!}{\input{tikz/range-tree.tikz}}
        \input{tikz/range-tree.tikz}
        \caption{\label{fig:range-tree} An illustration of the range tree
            consisting of a single range from ASK ($u$) to the key before LIME
            ($n_u - 1$).
            Nodes along the left and right paths, $p_l$ and $p_r$, are denoted
            as squares, while their lowest common ancestor $a$ is denoted as a
            diamond.
            Shaded nodes in blue and orange indicate the right or left children
            of $p_l$ and $p_r$ respectively, which are set to 1.
            The key at index $n_u - 1$, shown in purple, is also set to 1.
        }
        \vspace*{-\medskipamount}
    \end{figure}
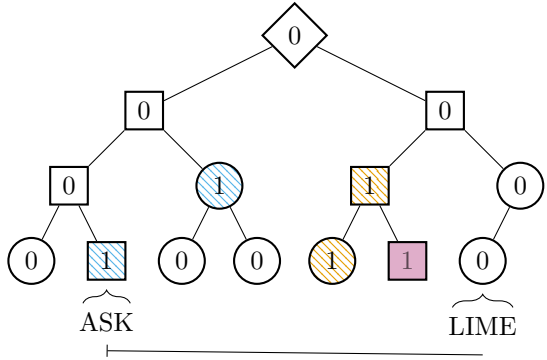

	\begin{fact}
        \vspace*{-\smallskipamount}
        The union of the leaves descending from `set' nodes in a range
        tree is the union of the original ranges.
        \vspace*{-\smallskipamount}
    \end{fact}

	Each node that is a candidate can then, concurrently, check all of its
	ancestors in the range tree, and if any of them are set to 1, eliminate
	itself as a candidate.
	This process is done with $\mathcal{O}(B \log{B})$ work and
	$\mathcal{O}(1)$ span.

	\Cref{step:pb-c} is done by a single LSW (least significant word) operation,
	which can be done as the reverse of the MSW operation as described in
	\Cref{sec:pram-string-comparisons}, and which takes constant span and
	$\mathcal{O}(B)$ work.
	The node at index 0 is always a candidate.

	The comparison of the search key $x$ with the key at index $w$ is done using
	the parallel \textsc{k-Compare} procedures as described in
	\Cref{sec:parallel-zt}, referring to the LCP value obtained as $\ell$.
	In \Cref{step:ba-3}, if $x$ is determined to be less than $k_w$, then we
	could determine the index $y$ by iterating to the left as $\ell_i \geq
	\ell$.
	This would, however, take up to $\mathcal{O}(B)$ span.
	Instead, each node at index $1 < i \leq w$ simultaneously determines if
	$\ell_i \geq \ell$ and if $\ell_{i-1} < \ell$.
	The right-most node that satisfies both conditions can be found with a
	single LSW operation, and is determined as the index $y$.
	This takes at most constant span and $\mathcal{O}(B)$ work.
	Similar steps can be taken if $x$ is determined to be greater than $k_w$.
\end{proof}

Since the number of string comparisons made during a string B-tree search $A(n) =
\Theta(\log_B{n})$, we can apply our parallelized \textsc{k-Compare} procedure
from \Cref{sec:parallel-zt} along with the branching algorithm to achieve the
following key results:
\begin{theorem} \label[theorem]{thm:string-b-tree-pram}
	By setting $B = \log{n}$, a parallel string B-tree can perform prefix search
	in $\mathcal{O}(\frac{\log{n}}{\log{\log{n}}})$ span and
	$\mathcal{O}(\frac{k}{\alpha} + \log^2{n})$ work in the practical CRCW PRAM
	model.
	Operations that return $m$ keys can be done in the same span and in
	$\mathcal{O}(m)$ additional work.
\end{theorem}

Since B-trees are designed primarily for external memory, we also provide the
results in the \emph{parallel external memory} (PEM)
model~\cite{argeFundamentalParallelAlgorithms2008}, where processors are allowed
to perform I/Os in parallel:

\begin{theorem} \label[theorem]{thm:string-b-tree-pem}
	A parallel string B-tree can perform prefix search in
	$\mathcal{O}(\log_B{n})$ I/O span and $\mathcal{O}(\frac{k}{\alpha B} + \log_B{n})$
	I/O work in the practical CRCW PEM model.
	Operations that return $m$ keys can be done in the same span and in
	$\mathcal{O}(m/B)$ additional I/O work.
\end{theorem}

The proofs, along with LCP-aware variants, are left to an appendix.

%% file: tikz/b-tree-branching.tikz
\begin{tikzpicture}[
	n/.style = { % node
        rectangle, draw, minimum width = 0.5cm, minimum height=0.4cm,
        font=\verytiny, inner sep=2pt, color=okabe1
    },
]
	\node[n] (AIM) at (0,0) {AIM};
	\node[n] (ANT) at ($(AIM) + (1.2, 0)$) {\keyvalue{ANT}{1}};
	\node[n] (APPLE) at ($(ANT) + (1.2, 0)$) {\keyvalue{APPLE}{1}};
	\node[n] (APPLY) at ($(APPLE) + (1.2, 0)$) {\keyvalue{APPLY}{4}};
	\node[n] (APTLY) at ($(APPLY) + (1.2, 0)$) {\keyvalue{APTLY}{2}};
	\node[n] (ASK) at ($(APTLY) + (1.2, 0)$) {\keyvalue{ASK}{1}};
	\node[n] (ASKS) at ($(ASK) + (1.2, 0)$) {\keyvalue{ASKS}{3}};
	\node[n] (ASLEEP) at ($(ASKS) + (1.2, 0)$) {\keyvalue{ASLEEP}{2}};
	\node[n] (ASSIST) at ($(ASLEEP) + (1.2, 0)$) {\keyvalue{ASSIST}{2}};
	\node[n] (LIME) at ($(ASSIST) + (1.2, 0)$) {\keyvalue{LIME}{0}};
	\node[n] (PAT) at ($(LIME) + (1.2, 0)$) {\keyvalue{PAT}{0}};
	\node[n] (PEA) at ($(PAT) + (1.2, 0)$) {\keyvalue{PEA}{1}};
	\node[n] (PEACH) at ($(PEA) + (1.2, 0)$) {\keyvalue{PEACH}{3}};
	\node[n] (PEAK) at ($(PEACH) + (1.2, 0)$) {\keyvalue{PEAK}{3}};

	\node[above,font=\tiny] at (ANT.north) {(1,\keyvalue{N}{0})};
	\node[above,font=\tiny] at (APPLE.north) {(1,\keyvalue{P}{0})};
	\node[above,font=\tiny] at (APPLY.north) {(4,\keyvalue{Y}{0})};
	\node[above,font=\tiny] at (APTLY.north) {(2,\keyvalue{T}{0})};
	\node[above,font=\tiny] at (ASK.north) {(1,\keyvalue{S}{0})};
	\node[above,font=\tiny] at (ASKS.north) {(3,\keyvalue{S}{0})};
	\node[above,font=\tiny] at (ASLEEP.north) {(2,\keyvalue{L}{0})};
	\node[above,font=\tiny] at (ASSIST.north) {(2,\keyvalue{S}{0})};
	\node[above,font=\tiny] at (LIME.north) {(0,\keyvalue{L}{0})};
	\node[above,font=\tiny] at (PAT.north) {(0,\keyvalue{P}{0})};
	\node[above,font=\tiny] at (PEA.north) {(1,\keyvalue{E}{0})};
	\node[above,font=\tiny] at (PEACH.north) {(3,\keyvalue{C}{0})};
	\node[above,font=\tiny] at (PEAK.north) {(3,\keyvalue{K}{0})};

	\draw[->] (AIM) -- (ANT);
	\draw[->] (ANT) -- (APPLE);
	\draw[->] (APPLE) -- (APPLY);
	\draw[->] (APPLY) -- (APTLY);
	\draw[->] (APTLY) -- (ASK);
	\draw[->] (ASK) -- (ASKS);
	\draw[->] (ASKS) -- (ASLEEP);
	\draw[->] (ASLEEP) -- (ASSIST);
	\draw[->] (ASSIST) -- (LIME);
	\draw[->] (LIME) -- (PAT);
	\draw[->] (PAT) -- (PEA);
	\draw[->] (PEA) -- (PEACH);
	\draw[->] (PEACH) -- (PEAK);
\end{tikzpicture}

%% file: tikz/b-tree-branching-pram.tikz
\begin{tikzpicture}[
	n/.style = { % node
        rectangle, draw, minimum width = 0.5cm, minimum height=0.4cm,
        font=\verytiny, inner sep=2pt, color=okabe1
    },
]
	\node[n] (AIM) at (0,0) {AIM};
	\node[n] (ANT) at ($(AIM) + (1.2, 0)$) {\keyvalue{ANT}{1}};
	\node[n] (APPLE) at ($(ANT) + (1.2, 0)$) {\keyvalue{APPLE}{1}};
	\node[n] (APPLY) at ($(APPLE) + (1.2, 0)$) {\keyvalue{APPLY}{4}};
	\node[n] (APTLY) at ($(APPLY) + (1.2, 0)$) {\keyvalue{APTLY}{2}};
	\node[n] (ASK) at ($(APTLY) + (1.2, 0)$) {\keyvalue{ASK}{1}};
	\node[n] (ASKS) at ($(ASK) + (1.2, 0)$) {\keyvalue{ASKS}{3}};
	\node[n] (ASLEEP) at ($(ASKS) + (1.2, 0)$) {\keyvalue{ASLEEP}{2}};
	\node[n] (ASSIST) at ($(ASLEEP) + (1.2, 0)$) {\keyvalue{ASSIST}{2}};
	\node[n] (LIME) at ($(ASSIST) + (1.2, 0)$) {\keyvalue{LIME}{0}};
	\node[n] (PAT) at ($(LIME) + (1.2, 0)$) {\keyvalue{PAT}{0}};
	\node[n] (PEA) at ($(PAT) + (1.2, 0)$) {\keyvalue{PEA}{1}};
	\node[n] (PEACH) at ($(PEA) + (1.2, 0)$) {\keyvalue{PEACH}{3}};
	\node[n] (PEAK) at ($(PEACH) + (1.2, 0)$) {\keyvalue{PEAK}{3}};

	\node[above,font=\tiny] at (ANT.north) {(1,\keyvalue{N}{0})};
	\node[above,font=\tiny] at (APPLE.north) {(1,\keyvalue{P}{0})};
	\node[above,font=\tiny] at (APPLY.north) {(4,\keyvalue{Y}{0})};
	\node[above,font=\tiny] at (APTLY.north) {(2,\keyvalue{T}{0})};
	\node[above,font=\tiny] at (ASK.north) {(1,\keyvalue{S}{0})};
	\node[above,font=\tiny] at (ASKS.north) {(3,\keyvalue{S}{0})};
	\node[above,font=\tiny] at (ASLEEP.north) {(2,\keyvalue{L}{0})};
	\node[above,font=\tiny] at (ASSIST.north) {(2,\keyvalue{S}{0})};
	\node[above,font=\tiny] at (LIME.north) {(0,\keyvalue{L}{0})};
	\node[above,font=\tiny] at (PAT.north) {(0,\keyvalue{P}{0})};
	\node[above,font=\tiny] at (PEA.north) {(1,\keyvalue{E}{0})};
	\node[above,font=\tiny] at (PEACH.north) {(3,\keyvalue{C}{0})};
	\node[above,font=\tiny] at (PEAK.north) {(3,\keyvalue{K}{0})};

	\draw[->] (AIM) -- (ANT);
	\draw[->] (ANT) -- (APPLE);
	\draw[->] (APPLE) -- (APPLY);
	\draw[->] (APPLY) -- (APTLY);
	\draw[->] (APTLY) -- (ASK);
	\draw[->] (ASK) -- (ASKS);
	\draw[->] (ASKS) -- (ASLEEP);
	\draw[->] (ASLEEP) -- (ASSIST);
	\draw[->] (ASSIST) -- (LIME);
	\draw[->] (LIME) -- (PAT);
	\draw[->] (PAT) -- (PEA);
	\draw[->] (PEA) -- (PEACH);
	\draw[->] (PEACH) -- (PEAK);

	\node[below] (bAIM1) at (AIM.south) {a)};

	\node[below] (bANT1) at (ANT.south) {$\leq$ \cmark};
	\node[below] (bAPPLE1) at (APPLE.south) {$\leq$ \cmark};
	\node[below] (bAPPLY1) at (APPLY.south) {$>$ \xmark};
	\node[below] (bAPTLY1) at (APTLY.south) {$>$ \xmark};
	\node[below] (bASK1) at (ASK.south) {$>$ \xmark};
	\node[below] (bASKS1) at (ASKS.south) {$>$ \xmark};
	\node[below] (bASLEEP1) at (ASLEEP.south) {$\leq$ \cmark};
	\node[below] (bASSIST1) at (ASSIST.south) {$>$ \xmark};
	\node[below] (bLIME1) at (LIME.south) {$>$ \xmark};
	\node[below] (bPAT1) at (PAT.south) {$>$ \xmark};
	\node[below] (bPEA1) at (PEA.south) {$\leq$ \cmark};
	\node[below] (bPEACH1) at (PEACH.south) {$\leq$ \cmark};
	\node[below] (bPEAK1) at (PEAK.south) {$>$ \xmark};

	\node[below] (bAIM2) at (bAIM1.south) {b)};

	\node[below] (bAIM3) at (bAIM2.south) {c)};

	\node (bANT3) at ($(bAIM3) + (1.2, 0)$) {$\leq$ \cmark};
	\node[rectangle,draw] (bAPPLE3) at ($(bANT3) + (1.2, 0)$) {$\leq$ \cmark};
	\node (bAPPLY3) at ($(bAPPLE3) + (1.2, 0)$) {$>$ \xmark};
	\node (bAPTLY3) at ($(bAPPLY3) + (1.2, 0)$) {$>$ \xmark};
	\node (bASK3) at ($(bAPTLY3) + (1.2, 0)$) {$>$ \xmark};
	\node (bASKS3) at ($(bASK3) + (1.2, 0)$) {$>$ \xmark};
	\node (bASLEEP3) at ($(bASKS3) + (1.2, 0)$) {$\leq$ \xmark};
	\node (bASSIST3) at ($(bASLEEP3) + (1.2, 0)$) {$>$ \xmark};
	\node (bLIME3) at ($(bASSIST3) + (1.2, 0)$) {$>$ \xmark};
	\node (bPAT3) at ($(bLIME3) + (1.2, 0)$) {$>$ \xmark};
	\node (bPEA3) at ($(bPAT3) + (1.2, 0)$) {$\leq$ \xmark};
	\node (bPEACH3) at ($(bPEA3) + (1.2, 0)$) {$\leq$ \xmark};
	\node (bPEAK3) at ($(bPEACH3) + (1.2, 0)$) {$>$ \xmark};

	% midpoint is 10.5pt, so 21pt is the full height. one fourth of that is 5.25pt

	\draw[|-] ([yshift=-4pt]bAPPLY1.south) -- ([yshift=-4pt]bAPTLY1.south);
	\draw[|-] ([yshift=+4pt]bAPTLY3.north) -- ([yshift=+4pt]bASK3.north);
	\draw[|-] ([yshift=-4pt]bASK1.south) -- ([yshift=-4pt]bLIME1.south);
	\draw[|-] ([yshift=+4pt]bASKS3.north) -- ([yshift=+4pt]bASSIST3.north);
	\draw[|-] ([yshift=+10.5pt]bASSIST3.north) -- ([yshift=+10.5pt]bLIME3.north);
	\draw[|-] ([yshift=+4pt]bLIME3.north) -- ([yshift=+4pt]bPAT3.north);
	\draw[|-] ([yshift=-4pt]bPAT1.south) -- ([yshift=-4pt,xshift=+0.3cm]bPEAK1.south);
	\draw[|-] ([yshift=+4pt]bPEAK3.north) -- ([yshift=+4pt,xshift=+0.3cm]bPEAK3.north);

	% \node[below] (bAIM4) at (bAIM3.south) {3.};
\end{tikzpicture}

%% file: tikz/range-tree.tikz
\begin{tikzpicture}[
    % edge from parent/.append style={->, >=latex, thick},
    n/.style={draw, circle, thick, minimum size=0.5cm},
    t/.style={draw, rectangle, thick, minimum size=0.5cm},
    a/.style={draw, diamond, thick, minimum size=0.5cm},
    level distance=1cm,
    level 1/.style={sibling distance=4cm},
    level 2/.style={sibling distance=2cm},
    level 3/.style={sibling distance=1cm},
	b/.style={pattern color=okabe3, pattern=north west lines},
	o/.style={pattern color=okabe2, pattern=north west lines},
    % c/.style={dashed,->,>=latex}
]
    \node[a] {0}
        child { node[t] {0}
            child { node[t] {0}
                child { node[n] {0} }
                child { node[t,b] (ex0) {1} }
            }
            child { node[n, b] {1}
                child { node[n] {0} }
                child { node[n] {0} }
            }
        }
        child { node[t] {0}
            child { node[t, o] {1}
                child { node[n, o] {1} }
                child { node[t, fill=okabe8, fill opacity=0.6] {1} }
            }
            child { node[n] {0}
                child { node[n] (ex1) {0} }
                child { node[draw=none] {} edge from parent[draw=none] }
            }
        }
    ;

    \node (ex0bits) [below=3mm of ex0] {ASK}; 
    \drawReducedWidthBrace{ex0bits}{above}{}

    \node (ex1bits) [below=3mm of ex1] {LIME};
    \drawReducedWidthBrace{ex1bits}{above}{}

	\draw[|-] ([yshift=-4pt]ex0bits.south) -- ([yshift=-4pt]ex1bits.south);
\end{tikzpicture}

%% file: experiments.tex
\section{Experiments}

We support our claims with experimental results, comparing the performance of
both our sequential and LCP-aware parallel implementations of zip-tries (which
we refer to as ZT and PZT, respectively), along with the memory-intensive
variants (which we refer to as MI-ZT and MI-PZT, respectively) with a
state-of-the-art string dictionary data structure,
\texttt{c-trie++}~\cite{tsurutaCtrieDynamicTrie2022}.
In a recent paper, \texttt{c-trie++} was shown to outperform two non-compact
trie data structures and four compact trie data structures, including the z-fast
trie~\cite{belazzouguiDynamicZfastTries2010} and the packed compact
trie~\cite{DBLP:journals/ieicet/TakagiISA17} in a majority of performed
experiments, performing particularly well in the case of long strings.
We used the ABC HuMi dataset~\cite{hirschABCHuMiAtlasBiosynthetic2024}, a
contemporary dataset of biosynthetic gene clusters, since the
genes are long but share only short common prefixes.
% each individual gene
% has a long median length but 
% 
% of around 18,000 nucleotides, but only a median LCP
% length with other genes of around 7.
The median gene length is around 18,000 nucleotides, but the median LCP length with
other genes is only around 7.
% Our code will be made available online.
Our code is available online at \url{https://github.com/ofekih/ZipAndSkipTries}.

% \subsection{Dependence on LCP Length}

% We first consider the time it takes to search for a randomized set of DNA
% strings in each data structure.
% Let's first consider the prefix search operation.
% Since all data structures are LCP-aware, we plot the time against the LCP
% length ($\ell$) of the searched strings.
% See \Cref{fig:lcp-length} (left).
In \Cref{fig:search-lcp-length}, we compare the time it takes to search for
strings sharing LCP length $\ell$ in each data structure.
These empirical results clearly show that even for shorter LCP lengths, our
sequential zip-trie variants (ZT and MI-ZT) consistently outperform the vastly
more complex \texttt{c-trie++}.
Meanwhile, the parallel variants (PZT and MI-PZT), while having by far the best
slopes, are overwhelmed by the parallel overhead, performing worse than all
other tested data structures on the dataset.
These results motivate a hybrid data structure which alternates between
CPU and GPU LCP-length calculations depending on the input size, but such a data
structure is beyond the scope of this paper.
% 
% We next consider the time it takes to insert a randomized set of DNA strings
% into each data structure.
% Next, we consider the insertion operation.
In \Cref{fig:construction-lcp-length}, we compare the time it takes to insert strings sharing
total LCP length $L$ into each data structure.
The empirical results show that, as with search operations, the sequential
zip-trie variants consistently outperform \texttt{c-trie++} even in the case of
a small total LCP length.
Furthermore, the memory-intensive variant which stores more accurate LCP length
information outperforms the canonical variant by a constant factor, as expected
from \Cref{lemma:approx-lcp-time}, confirming that it is the use of this
metadata that allows zip-tries to avoid a vast majority of string comparisons.
Unlike for the search operations, the vast majority of insertions involve
strings that share a very short LCP with others.
This hits the worst case for the parallel variants which incur a significant
overhead of GPU invocations despite reaping minimal benefits from the
parallelism.
This further motivates a hybrid data structure
as described earlier.

% \begin{figure}[t!]
% 	\centering
% 	% \hspace*{-18pt}
% 	% \begin{minipage}{1.1\textwidth}
% 	\includegraphics[width=\linewidth,trim={1cm 0 1cm 1cm}, clip]{figures/search-time-lcp-length-comparison.png}
% 	% \hspace*{-24pt}
% 	\includegraphics[width=\linewidth,trim={1cm 0 1cm 1cm}, clip]{figures/construction-time-total-LCP-length-comparison.png}
% 	% \end{minipage}
% 	% \vspace*{-\bigskipamount}
% 	\caption{\label{fig:lcp-length} Performance comparisons of
% 	our data structures against
% 	\texttt{c-trie++} in terms of the LCP length ($\ell$) or total LCP length
% 	($L$) of the searched (left) or inserted (right) keys.
% 	For search operations (left),  while the parallel variants have by far the best slope, their
% 	parallel overhead significantly impairs their performance.
% 	}
% 	\vspace*{-\bigskipamount}
% \end{figure}

\begin{figure}[ht!]
    \centering
    \includegraphics[width=\linewidth,trim={1cm 0 1cm 1cm}, clip]{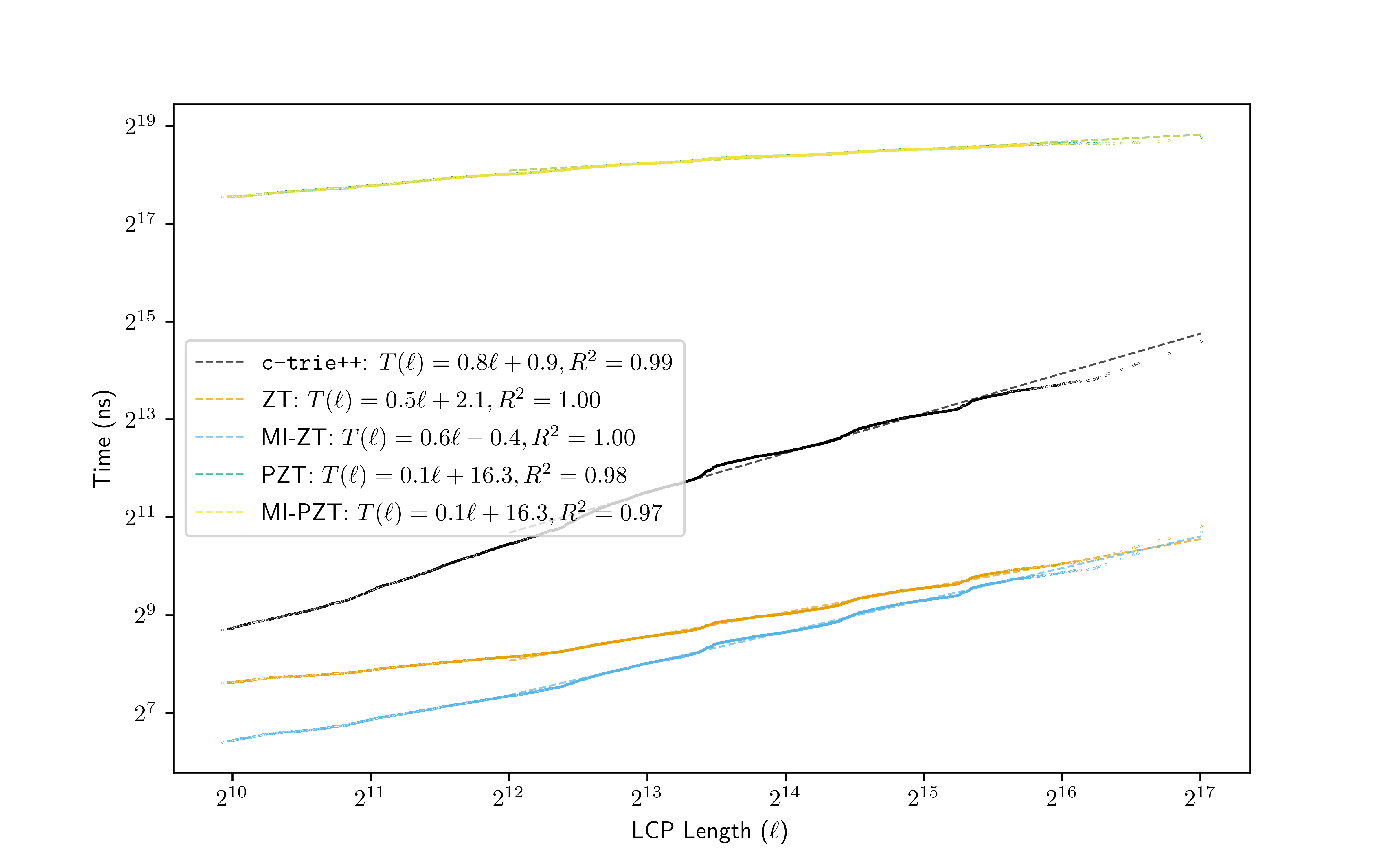}
    \caption{
		Performance comparison for the prefix search operation
		of our data structures against \texttt{c-trie++}
		in terms of the LCP length ($\ell$) of the searched keys.
		While the parallel variants have by far the best slope, their parallel
		overhead significantly impairs their performance.
		Our sequential algorithms outperform \texttt{c-trie++} even for shorter
		LCP lengths.
	}
    \label{fig:search-lcp-length}
\end{figure}

\begin{figure}[ht!]
    \centering
    \includegraphics[width=\linewidth,trim={1cm 0 1cm 1cm}, clip]{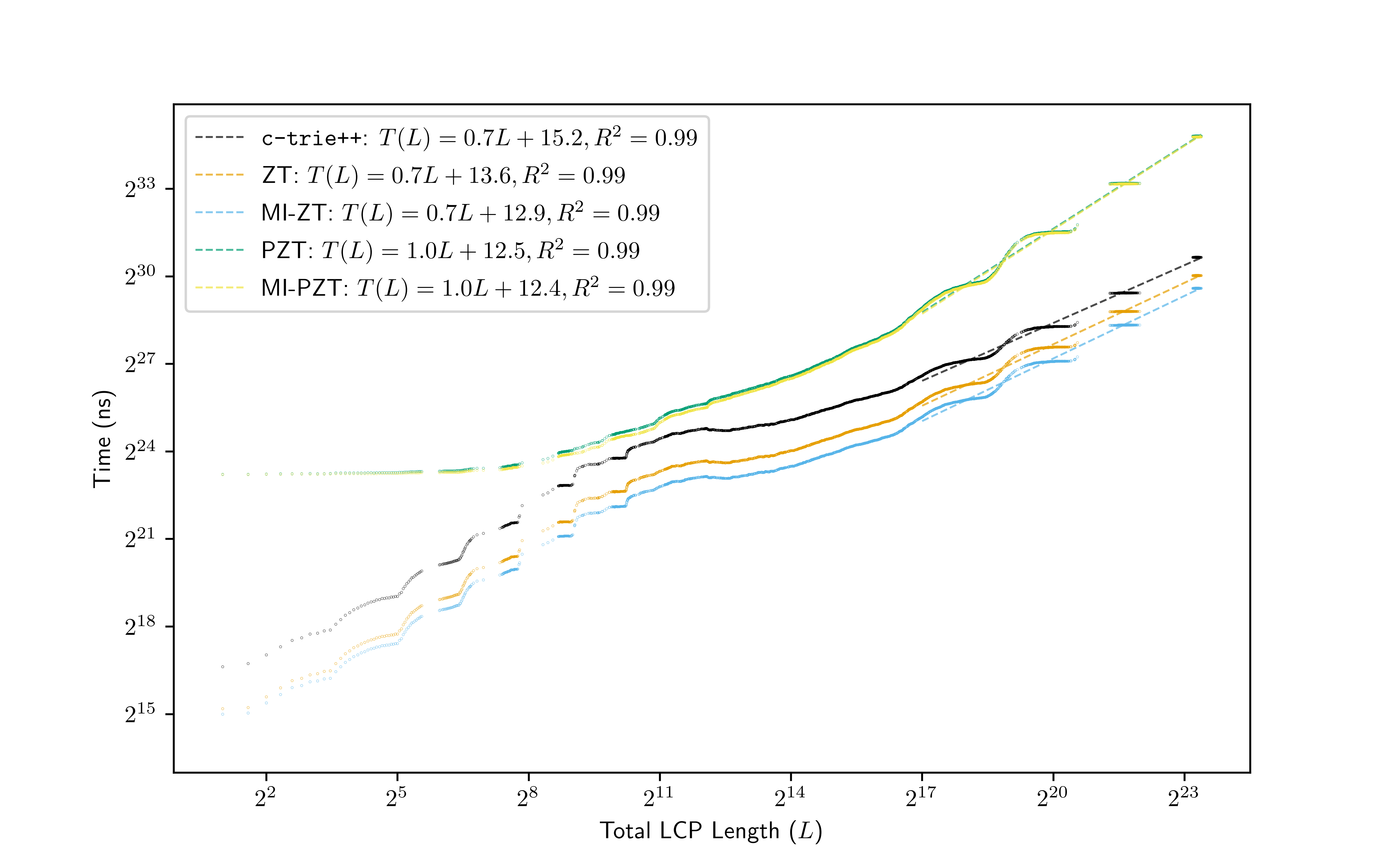}
    % \vspace*{-\bigskipamount} % Optional: Adjust vertical spacing if needed
	\caption{
		Performance comparison for the insertion operation
		of our data structures against \texttt{c-trie++}
		in terms of the total LCP length ($L$) of the inserted keys.
		While the parallel variants initially have the best slope, they are
		outperformed by the sequential algorithms.
		Our sequential algorithms perform competitively with \texttt{c-trie++},
		appearing to slightly outperform it.
	}
	\label{fig:construction-lcp-length}
\end{figure}

We include more detail on the experimental framework along with additional
results on how the performance scales with the number of keys $n$ in
an appendix.

%% file: future.tex
\section{Future Work}

It could be interesting to see whether our techniques could be applied to
update operations in the B-skip-trie.
There is a gap between the $\mathcal{O}(B)$ work required for the string
B-tree's branching operation in the RAM model and the $\mathcal{O}(B
\log{B})$-work algorithm of the B-skip-trie we present for the PRAM model which
would be interesting to see if it could be closed, potentially requiring
different metadata.
While we tested the skip-trie on large strings, it would also be interesting to
see how it compares to other trie variants on different types of data.
It would also be interesting to see if the skip-trie could be adapted to allow
concurrent updates, particularly in a way that preserves the edge values, as was
done for the skip list~\cite{GABARRO19961,pughConcurrentMaintenanceSkip1998}
and other trie variants~\cite{prokopecConcurrentTriesEfficient2012}.
Finally common trie variants apply prefix-based compression techniques to
greatly reduce the memory usage.
Our techniques should apply to these compressed strings as well, although more
experimentation would be necessary.

%% file: related-work.tex
\section{Related Prior Work}

There are many trie variants that perform well sequentially, but because of
rigid branching structures they cannot take advantage of bit-parallel
operations, or cannot be easily parallelized.
For example, the lexicographic splay trees by Sleator and
Tarjan~\cite{sleatorSelfadjustingBinarySearch1985} achieve $\mathcal{O}(k +
\log{n})$ update/query time, but do not take advantage of bit-parallel
operations.
The exponential search trees by Andersson and
Thorup~\cite{DBLP:journals/jacm/AnderssonT07} achieve $\mathcal{O}(\ell +
\sqrt{\frac{\log{n}}{\log{\log{n}}}})$ time, but are not easily parallelized.
The dynamic z-fast tries by Belazzougui
et~al.~\cite{belazzouguiDynamicZfastTries2010} achieve $\mathcal{O}(k/\alpha +
\log{m})$ time, returning a correct result w.h.p., but are complex and not
easily parallelized.
Some more recent variants include the wexponential search trees of Fischer and
Gawrychowski~\cite{DBLP:conf/cpm/0001G15} which, when implemented dynamically,
achieve $\mathcal{O}(k + \frac{\log^2{\log{\sigma}}}{\log{\log{\log{\sigma}}}})$
time, and the packed compact tries of Takagi
et~al.~\cite{DBLP:journals/ieicet/TakagiISA17} which achieve
$\mathcal{O}(k/\alpha + \log{\log{n}})$ time.
These methods also are complex and not easily parallelized.

In contrast, skip
lists~\cite{pughSkipListsProbabilistic1990,munroDeterministicSkipLists1992} and
zip-trees~\cite{tarjanZipTrees2021a,gilaZipzipTreesMaking2023} are simple data
structures, but they are almost always defined for keys that are either
numerical or in contexts where key comparisons are fast.
Pugh, the original inventor of the skip list data structure, discussed the
problem of expensive comparisons, including for strings,
in~\cite{pughSkipListCookbook1990}, but his proposed solution does not yield any
asymptotic improvement.
Grossi and Italiano created a general framework for adapting linked data
structures, such as skip lists and binary search trees, to support
multidimensional keys such as strings~\cite{10.1007/3-540-48523-6_34}.
Their framework when applied to unbalanced binary search trees and to AVL trees
was shown to be competitive with the state of the art string dictionary data
structures at the time~\cite{10.1007/3-540-44867-5_7}.
Irving and Love showed how to simplify the framework specifically for the case
of suffix trees by removing the need for parent pointers~\cite{IRVING2003387}.
Both frameworks require exponentially more memory usage than is desired for
zip-trees when applied out-of-the-box.

Tries have been discussed in the context of multiple processors for several use
cases, such as for concurrent
maintenance~\cite{prokopecConcurrentTriesEfficient2012}, for parallel string
matching~\cite{galil*OptimalParallelAlgorithms1984,gasieniecWorktimeOptimalParallel1994},
and for the multistring search
problem~\cite{ferraginaStringSearchCoarseGrained1999}.
There is a more recent result by Jekovec and
Brodnik~\cite{jekovecParallelQuerySuffix2015} which achieves good results for a
static suffix tree data structure under the small PRAM model, where the amount
of processors $p \ll k$.

Several specialized string dictionary data structures that involve fewer
branching paths have been proposed for the external memory model, namely the
self-adjusting skip list (SASL) of Ciriani et
al.~\cite{cirianiDataStructureSequence2007}, and the string B-tree of Ferragina
and Grossi~\cite{ferraginaStringBtreeNew1999}.
While we show how to apply our parallelization techniques to the string B-tree
data structure to achieve great PRAM results, we note that these data structures
are far more complex than zip-tries and are largely impractical in
practice~\cite{ferraginaStringAlgorithmsData2008}, and are consequently not
competitive with the current state-of-the-art dynamic string dictionary data
structures~\cite{tsurutaCtrieDynamicTrie2022}.
There exists a cache-oblivious variant of the string B-tree by Bender et
al.~\cite{benderCacheobliviousStringBtrees2006}.
String B-trees have been implemented both
statically~\cite{fanImplementationEvaluationString2004,martinez-prietoPracticalCompressedString2016}
and
dynamically~\cite{ferraginaFastStringSearching1996,joo-youngEffectiveImplementationString2006,wuDistributedTrueString2006}.
There is also a more recent string matching data structure based on the B-tree,
called the DMP tree~\cite{yazdaniDMPtreeDynamicMway2010}, but without any
theoretical guarantees.
To the extent of our knowledge, no variant has been adapted for the PEM and PRAM
settings as we have done in this paper.

With respect to more memory-efficient practical implementations, there is a C++
library for implementing dynamic compressed string processing data structures
which may make such data structures more space efficient at the cost of
time~\cite{DBLP:conf/wea/Prezza17}.
The best previous sequential trie library,
\texttt{c-trie++}~\cite{tsurutaCtrieDynamicTrie2022}, is an optimized
implementation of the z-fast trie data structure.
We found that our zip-trie data structure, despite being significantly simpler,
outperforms the \texttt{c-trie++} in terms of trie construction and search
operations when strings are large.
Moreover, \texttt{c-trie++} does not provide predecessor/successor queries,
prefix search, or range queries.
Furthermore, like other complex trie variants, \texttt{c-trie++} is not easily
parallelizable.

%% file: string-comparisons.tex
\section{String Comparisons}

A naive sequential algorithm which compares each letter of each word
individually would take $\mathcal{O}(\ell)$ time to compare two strings sharing
a longest common prefix of length $\ell$.
In this section, we review several more efficient algorithms under both
sequential and parallel models of computation.

\subsection{Sequential String Comparison Algorithms} \label{sec:seq-string-comparisons}

Under either the word RAM or the practical RAM models, comparisons between two
strings sharing a common prefix of length $\ell$ can be done in
$\mathcal{O}(\ell/\alpha)$ time, where $\alpha$ is the number of characters that
fit into a single machine word.
We briefly review one such algorithm below.

As is well known (e.g., see~\cite{thorup-ac0}), for two packed strings, the
first position at which they differ can be trivially determined by the most
significant bit (MSB) of their bitwise exclusive or (XOR).
XOR is a standard machine instruction; MSB is sometimes referred to as
count-left-zero (CLZ), and is also available as a primitive in most modern
instruction sets.
For concreteness we assume that the MSB of the zero word is the length of the
word in bits.
Once the longest common prefix $\ell$ is determined, we can simply compare the
next character in the strings to determine which is lexicographically greater,
as is done in the \textsc{k-Compare} algorithm of \Cref{sec:paradigm}.
(See~\Cref{fig:text-cmp}).
Thus, by using these instructions, we obtain a factor of $\alpha$ improvement
over the naive algorithm, achieving $\mathcal{O}(\ell/\alpha)$ string comparison
time.
The proof of~\Cref{cor:seq-compare-string} directly follows.

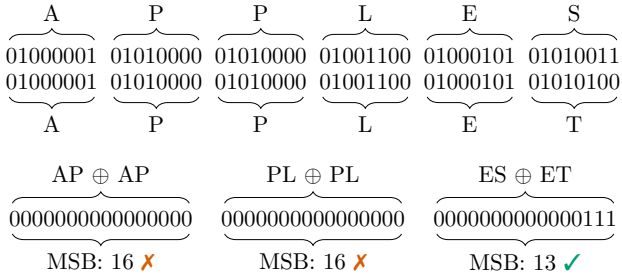
\begin{figure}[tb]
    \centering
    \resizebox{\linewidth}{!}{\input{tikz/comp.tikz}}
    \caption{\label{fig:text-cmp}Using bitwise XOR and MSB to identify the
    length of the longest common prefix between APPLES and APPLET represented in
    ASCII as it would run on a machine with word size 16. These operations
    revealed that the MSB of the XOR differed at the 13th bit of the 3rd machine
    word, corresponding the 6th character. Therefore, the LCP has length 5.
    %Note that modern machines have word sizes ranging from 64 bits to 512 bits
    %or more, allowing many more characters to be compared in parallel.
    }
\end{figure}

\subsection{String Comparisons in PRAM} \label{sec:pram-string-comparisons}

Under the practical PRAM model we assume that we can perform common operations
on each machine word in constant span and work.
With these assumptions, computing the LCP operation on two strings that each fit
into $M$ machine words can be reduced to three simple steps:
\begin{enumerate}
	\item Compute the bitwise XOR of the two strings. \label{step:1-non-zero}
	
	\item Determine which is the first non-zero word, also referred to as the
	most significant word (MSW). \label{step:2-msw}

	\item Determine the MSB of the most significant word. \label{step:3-msb}
\end{enumerate}

\crefname{enumi}{step}{steps}
\Crefname{enumi}{Step}{Steps}

\Cref{step:1-non-zero} is embarrassingly parallel and under the practical PRAM
model can be done in constant span and in $\mathcal{O}(M)$ work.
\Cref{step:3-msb} can also trivially be done in constant span and work.
Fich, Ragde, and Wigderson show that \Cref{step:2-msw} can be done using only
$\sqrt{M}$ cells of memory in constant span and $\mathcal{O}(M)$ work in the
common CRCW PRAM model, where processors can only concurrently write to the same
memory location if they are writing the same value~\cite{doi:10.1137/0217037}.

Restated briefly as it applies to our problem, it is trivial to compute the MSW
of $M$ machine words in constant span and $\mathcal{O}(M^2)$ work by checking
each of the $\binom{M}{2}$ possible pairs of words simultaneously, and setting
the second machine word to 0 if the first is non-zero.
Consequently, we can compute the MSW of $\sqrt{M}$ machine words in constant
span and $\mathcal{O}(M)$ work.
As shown in \Cref{fig:cmp-sqrt}, we can first reduce the problem to one of size
$\sqrt{M}$ by grouping each $\sqrt{M}$ consecutive machine words together into
one representative machine word that is set to 1 if any of the corresponding
words are non-zero (and 0 otherwise).
Then, we can compute the MSW of the reduced problem of size $\sqrt{M}$.
This reduced problem will point us to the correct block of size $\sqrt{M}$
machine words in the original problem, which we can compute with a second MSW
operation.
Since each of these steps can be done in constant span and linear work, so can
the entire operation.

\begin{figure}[ht!]
    \centering
    \resizebox{0.8\linewidth}{!}{\input{tikz/comp-sqrt.tikz}}
    % \resizebox{\linewidth}{!}{\input{tikz/comp-sqrt.tikz}}
    % \input{tikz/comp-sqrt.tikz}
    % \vspace*{-\bigskipamount}
    \caption{\label{fig:cmp-sqrt}Reducing a MSW problem of size $M$ to one of
    size $\sqrt{M}$, where each machine word in the reduced problem corresponds
    to $\sqrt{M}$ machine words in the original problem.}
% \vspace*{-6pt}
\end{figure}
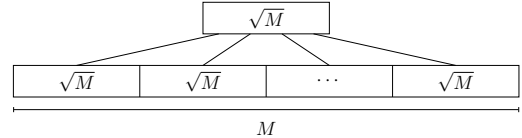

These results directly lead to the following lemma.

\begin{lemma}\label{lem:parallel-string-compare}
    We can compute the MSB of a
    binary string stored in in $M$ machine words in $\mathcal{O}(1)$ span and
    $\mathcal{O}(M)$ work under the common practical CRCW
    PRAM model.
    Consequently, two $k$-length strings can be compared in $\mathcal{O}(1)$
    span and $\mathcal{O}(k/\alpha)$ work.
\end{lemma}

\subsection{String Comparisons in PEM} \label{sec:pem-string-comparisons}

In this section we take advantage of how strings are stored in contiguous memory
to prove optimal bounds in the parallel external memory (PEM) model where
processors can concurrently transfer blocks of size $B$ to and from main memory.

Similar to before, we consider each machine block as significant if it contains
any significant (non-zero) word.
The most significant block can be found in a similar fashion as the MSW from
\Cref{step:2-msw}.
Since there are only $M/B$ blocks, we can compute the most significant block in
constant span and $\mathcal{O}(M/B)$ work.
Finally, we can trivially determine the MSB of the most significant block using
only one I/O.
We obtain a similar result as before.

\begin{lemma} \label{lem:pem-string-compare} We can compute the MSB of a binary
	string stored in $M$ machine words in $\mathcal{O}(1)$ I/O span and
	$\mathcal{O}(M/B)$ I/O work under the common practical CRCW PEM model.
	Consequently, two $k$-length strings can be compared in $\mathcal{O}(1)$
	I/O span and $\mathcal{O}(\frac{k}{\alpha B})$ I/O work.
\end{lemma}

%% file: tikz/comp.tikz
\begin{tikzpicture}
    \node (APPLESa)  at (-3, 0)    {01000001};
    \node (APPLESp1) at (-1.4, 0) {01010000};
    \node (APPLESp2) at (0.2, 0)    {01010000};
    \node (APPLESl)  at (1.8, 0)   {01001100};
    \node (APPLESe)  at (3.4, 0)     {01000101};
    \node (APPLESs)  at (5, 0)   {01010011};

    \node (APPLETa)  at (-3, -0.4)   {01000001};
    \node (APPLETp1) at (-1.4, -0.4) {01010000};
    \node (APPLETp2) at (0.2, -0.4)    {01010000};
    \node (APPLETl)  at (1.8, -0.4)  {01001100};
    \node (APPLETe)  at (3.4, -0.4)    {01000101};
    \node (APPLETt)  at (5, -0.4)  {01010100};

    \node (APxorAP) at (-2.25, -2.5) {0000000000000000};
    \node (PLxorPL) at (1, -2.5)  {0000000000000000};
    \node (ESxorET) at (4.25, -2.5)  {0000000000000111};

    \drawReducedWidthBrace{APPLESa}{above}{A}
    \drawReducedWidthBrace{APPLESp1}{above}{P}
    \drawReducedWidthBrace{APPLESp2}{above}{P}
    \drawReducedWidthBrace{APPLESl}{above}{L}
    \drawReducedWidthBrace{APPLESe}{above}{E}
    \drawReducedWidthBrace{APPLESs}{above}{S}

    \drawReducedWidthBrace{APPLETa}{below}{A}
    \drawReducedWidthBrace{APPLETp1}{below}{P}
    \drawReducedWidthBrace{APPLETp2}{below}{P}
    \drawReducedWidthBrace{APPLETl}{below}{L}
    \drawReducedWidthBrace{APPLETe}{below}{E}
    \drawReducedWidthBrace{APPLETt}{below}{T}

    \drawReducedWidthBrace{APxorAP}{above}{AP $\oplus$ AP}
    \drawReducedWidthBrace{PLxorPL}{above}{PL $\oplus$ PL}
    \drawReducedWidthBrace{ESxorET}{above}{ES $\oplus$ ET}

    \drawReducedWidthBrace{APxorAP}{below}{MSB: 16 \xmark}
    \drawReducedWidthBrace{PLxorPL}{below}{MSB: 16 \xmark}
    \drawReducedWidthBrace{ESxorET}{below}{MSB: 13 \cmark}
\end{tikzpicture}

%% file: tikz/comp-sqrt.tikz
\begin{tikzpicture}[
	font=\LARGE,
]
	% draw rectangle
	\draw[color=okabe1] (0,0) rectangle (16,1);

	% draw |-| line beneath this rectangle, with the text $M$ below it
	% the style should be |-|, and there should be a midway node
	\draw[color=okabe1,|-|] (0,-0.4) -- (16,-0.4) node[midway,below=0.25] {$M$};

	% draw vertical lines at 4, 8, 12
	\draw[color=okabe1] (4,0) -- (4,1);
	\draw[color=okabe1] (8,0) -- (8,1);
	\draw[color=okabe1] (12,0) -- (12,1);

	% write sqrt(M) in each box
	\node at (2,0.5) {$\sqrt{M}$};
	\node at (6,0.5) {$\sqrt{M}$};
	\node at (10,0.5) {$\cdots$};
	\node at (14,0.5) {$\sqrt{M}$};

	\draw[color=okabe1] (6,2) rectangle (10,3);

	% write sqrt M
	\node at (8,2.5) {$\sqrt{M}$};

	% draw lines from each small rectangle to corresponding spot in the top one
	\draw[color=okabe1] (2,1) -- (6.5,2);
	\draw[color=okabe1] (6,1) -- (7.5,2);
	\draw[color=okabe1] (10,1) -- (8.5,2);
	\draw[color=okabe1] (14,1) -- (9.5,2);
\end{tikzpicture}

%% file: more-experiments.tex
\section{Supplemental Experiments} \label{sec:more-experiments}

In this section we include some more details regarding our experimental
framework and include results for how the data structures scale with the number
of keys $n$.

All data structures were implemented in C++ and using CUDA and compiled with the same
optimizations.
Our experiments were run on a machine with an Intel Core i7-8750H CPU and a 1280-core
NVIDIA GTX 1060 Mobile GPU.
The \texttt{c-trie++} library converts each input string into a custom
\texttt{LongString} type.
This conversion along with some preprocessing 
was done before the start of each experiment.
The parallel implementations perform LCP length calculations fully in parallel,
meaning that they invoke the GPU even for very small inputs.
In practice, only invoking the GPU after some threshold input size may result in
significantly more practical implementations.

As stated in the main text, we are primarily interested in the case where the
length of the LCP ($\ell$) between strings is much larger than the total number
of strings ($n$), specifically where $\ell \gg \log{n}$.
This was primarily the case for prefix search operations sharing larger values
of $\ell$, as shown in \Cref{fig:search-lcp-length}.
We justify this claim in \Cref{fig:search-num-keys}, which, when compared
with \Cref{fig:search-lcp-length}, shows that the time taken to perform search operations
was much more correlated with the LCP length of the strings than with the number
of strings in the data structure.
This was not the case for insertion operations for two key reasons.
First, we were inserting $n$ strings so we would regardless expect to see at
least a linear dependence on $n$.
Second, insertion operations specifically involve strings that are not already
in the data structure, and DNA strings in our dataset generally do not share
very long LCPs with each other.
Here we can start to see the advantage of \texttt{c-trie++}, which has a
theoretically better dependence on $n$ than our trie variants and consequently
we observe a better slope.
While both our sequential variants still outperform \texttt{c-trie++} over this
moderately sized dataset of around 18,000 keys, they will likely be overtaken
for significantly larger inputs.
See \Cref{fig:construction-num-keys}.

% \begin{figure}[ht!]
% 	\centering
% 	% \hspace*{-18pt}
% 	% \begin{minipage}{1.1\textwidth}
% 	\includegraphics[width=\linewidth,trim={1cm 0 1cm 1cm}, clip]{figures/search-time-num-keys-comparison.png}
% 	% \hspace*{-24pt}
% 	\includegraphics[width=\linewidth,trim={1cm 0 1cm 1cm}, clip]{figures/construction-time-num-keys-comparison.png}
% 	% \end{minipage}
% 	% \vspace*{-\bigskipamount}
% 	\caption{\label{fig:num-keys} Performance comparisons of the data
% 	structures in terms of the number of keys ($n$) in the data structure.
% 	The apparent lack of correlation between the number of keys and the time
% 	taken for search operations (left) confirms that, for the high-dimensional data we
% 	tested against, the number of keys $n$ was insignificant.
% 	}
% 	% \vspace*{-\medskipamount}
% \end{figure}

% split into two

\begin{figure}[ht!]
	\centering
	\includegraphics[width=\linewidth,trim={1cm 0 1cm 1cm}, clip]{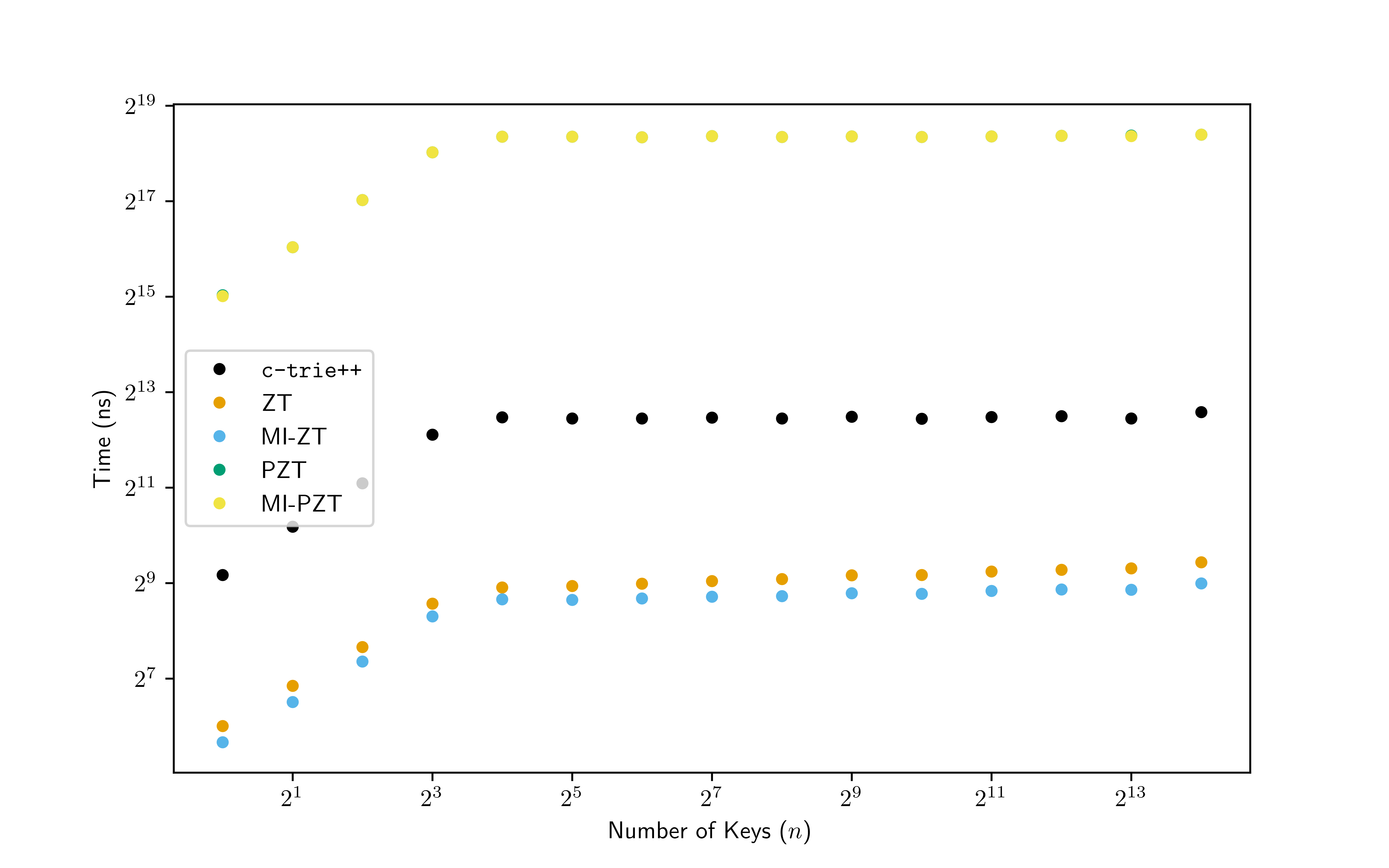}
	\caption{
		Performance comparison for the prefix search operation
		of our data structures against \texttt{c-trie++}
		in terms of the number of keys ($n$) in the data structure.
		The apparent lack of correlation between the number of keys and the time
		taken for search operations confirms that, for the high-dimensional data we
		tested against, the number of keys $n$ was insignificant.
	}
	\label{fig:search-num-keys}
\end{figure}

\begin{figure}[ht!]
	\centering
	\includegraphics[width=\linewidth,trim={1cm 0 1cm 1cm}, clip]{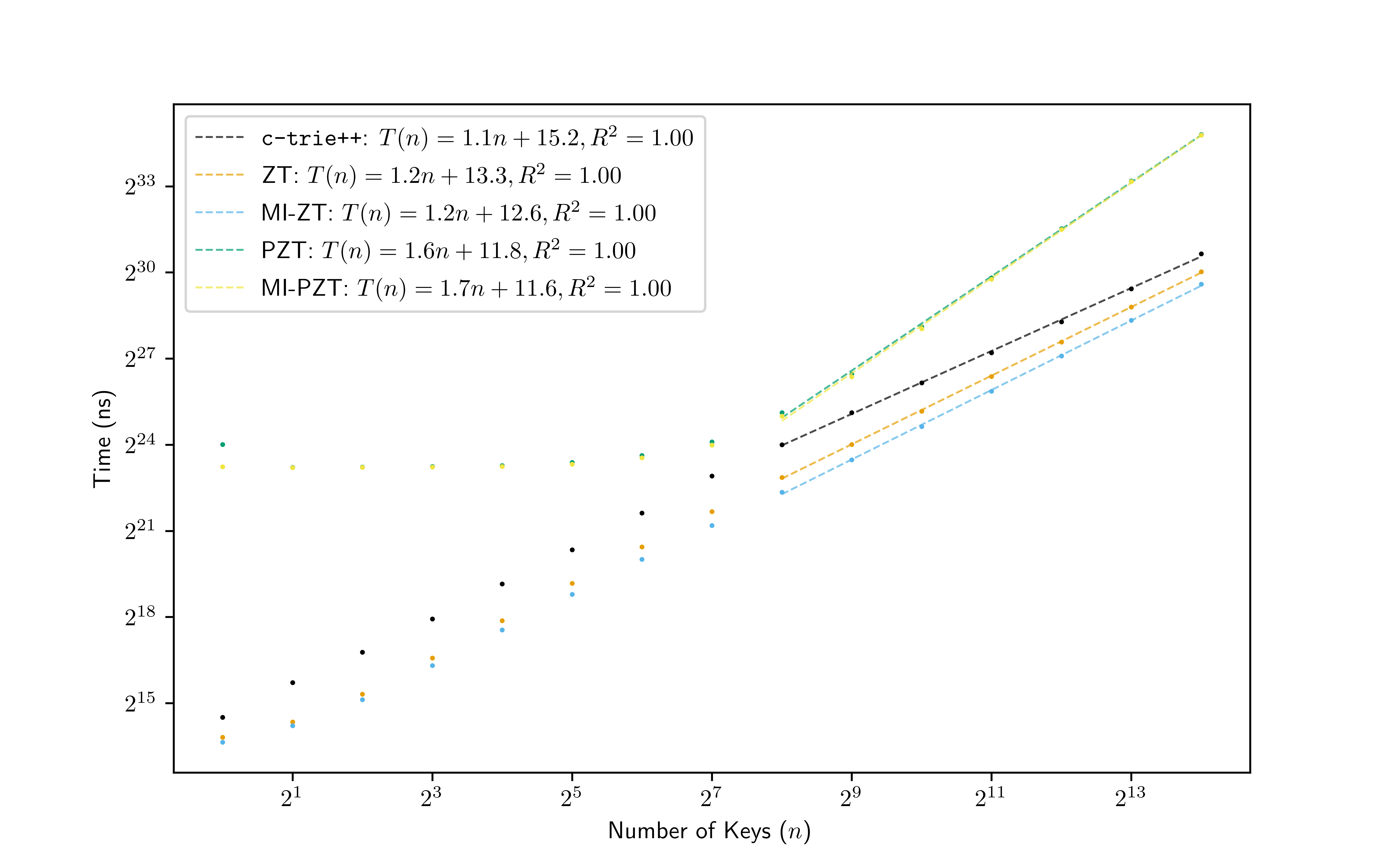}
	\caption{
		Performance comparison for the insertion operation
		of our data structures against \texttt{c-trie++}
		in terms of the number of keys ($n$) in the data structure.
		While the parallel variants initially have the best slope, they are
		outperformed by the sequential algorithms.
		Our sequential algorithms perform competitively with \texttt{c-trie++},
		although \texttt{c-trie++} appears to have the better slope, reflecting
		its better dependence on $n$.
	}
	\label{fig:construction-num-keys}
\end{figure}

%% file: omitted-proofs.tex
\section{Omitted Parallel String B-Tree Results}

In this appendix, we provide proofs for the results of the parallel string
B-tree algorithms that were omitted from the main text.

\begin{theorem}[Same as \Cref{thm:string-b-tree-pram}]
	By setting $B = \log{n}$, a parallel string B-tree can perform prefix search
	in $\mathcal{O}(\frac{\log{n}}{\log{\log{n}}})$ span and
	$\mathcal{O}(\frac{k}{\alpha} + \log^2{n})$ work in the practical CRCW PRAM
	model.
	Operations that return $m$ keys can be done in the same span and in
	$\mathcal{O}(m)$ additional work.
\end{theorem}

\begin{proof}
	Each node in the string B-tree has $\mathcal{O}(B)$ children, resulting in
	the height of the tree being $\mathcal{O}(\log_B{n})$.
	When $B = \log{n}$, the height of the tree, or the maximum number of nodes
	visited or comparisons made during a search, $A(n)$,
	is $\mathcal{O}(\log_{\log{n}}{n}) = \mathcal{O}(\frac{\log{n}}{\log{\log{n}}})$.
	Applying the results from \Cref{theorem:parallel-zt-k}, we obtain that the
	string B-tree can spend $\mathcal{O}(\frac{\log{n}}{\log{\log{n}}})$ span
	and $\mathcal{O}(k + \frac{\log{n}}{\log{\log{n}}})$ work on comparisons
	throughout the search.
	In addition, $\mathcal{O}(\frac{\log{n}}{\log{\log{n}}})$ nodes are visited,
	where, due to our branching algorithm results from \Cref{lem:branching}, we
	spend constant span and $\mathcal{O}(B \log{B}) = \mathcal{O}(\log{n}
	\log{\log{n}})$ work per node.
	Overall, we incur an additional $\mathcal{O}(\frac{\log{n}}{\log{\log{n}}})$
	span and $\mathcal{O}(\log^2{n})$ work due to node traversal and the
	branching algorithm.
	Combining these results along with packing $\alpha$ characters per word in
	the word RAM model, we obtain the desired result.

	Range queries can be done by first performing a prefix search on the
	boundary keys of the range, $u$ and $v$, keeping track of each key visited
	for each node in the search.
	After the search, we can determine the value $m$ and allocate an array
	sufficient to store $m$ pointers.
	We then spawn $\mathcal{O}(m/\log{n})$ processors which independently
	traverse down to corresponding leaf nodes between $u$ and $v$ in
	$\mathcal{O}(\frac{\log{n}}{\log{\log{n}}})$ span and work each.
	Finally, the $m$ threads can concurrently read the pointers of the $m$ keys
	between $u$ and $v$ and store the results in the corresponding array
	indices.
\end{proof}

\begin{corollary}
	Let $\ell$ be the length of the longest common prefix between a key $x$ and
	the stored keys in a parallel string B-tree $T$.
	By setting $B = \log{n}$, $T$ can perform prefix search
	in $\log{\ell} + \mathcal{O}(\frac{\log{n}}{\log{\log{n}}})$ span and
	$\mathcal{O}(\frac{\ell}{\alpha} + \log^2{n})$ work in the practical CRCW PRAM
	model.
	Operations that return $m$ keys can be done in the same span and in
	$\mathcal{O}(m)$ additional work.
\end{corollary}

The proof of this corollary directly follows from the previous proof using the
LCP-aware parallel \textsc{k-Compare} procedure from
\Cref{theorem:parallel-zt-l}.

\begin{theorem}[Same as \Cref{thm:string-b-tree-pem}]
	A parallel string B-tree can perform prefix search in
	$\mathcal{O}(\log_B{n})$ I/O span and $\mathcal{O}(\frac{k}{\alpha B} + \log_B{n})$
	I/O work in the practical CRCW PEM model.
	Operations returning $m$ keys can be done in the same span and in
	$\mathcal{O}(\frac{m}{B})$ additional work.
\end{theorem}

\begin{proof}
	It is possible to spend $\mathcal{O}(\log_B{n})$ I/O span and
	$\mathcal{O}(\frac{k}{\alpha B} + \log_B{n})$ I/O work on comparisons by
	combining the results of \Cref{theorem:parallel-zt-k} with the PEM-friendly
	string comparison oracle from \Cref{lem:pem-string-compare}.
	Any naive branching algorithm can be implemented using $\mathcal{O}(1)$ I/Os
	per node, obtaining the desired results.

	Range queries can be performed similarly to the PRAM model, spawning $k/B$
	processors instead of $k/\log{n}$ processors to traverse down to the leaf
	nodes between $u$ and $v$.
\end{proof}

\begin{corollary}
	Let $\ell$ be the length of the longest common prefix between a key $x$ and
	the stored keys in a parallel string B-tree $T$.
	$T$ can perform prefix search in $\log{\ell} + \mathcal{O}(\log_B{n})$ I/O
	span and $\mathcal{O}(\frac{\ell}{\alpha B} + \log_B{n})$ I/O work in the
	practical CRCW PEM model.
	Operations returning $m$ keys can be done in the same span and in
	$\mathcal{O}(\frac{m}{B})$ additional I/O work.
\end{corollary}

The proof of this corollary directly follows from the previous proof using the
LCP-aware parallel \textsc{k-Compare} procedure from
\Cref{theorem:parallel-zt-l}.